\documentclass[12pt]{article}

\usepackage{amsmath,amsfonts,amssymb,amstext}
\usepackage[overload]{empheq}

\usepackage[normal]{subfigure}
\usepackage{epsfig}
\usepackage{graphicx}
\usepackage{enumerate}
\usepackage{mathtools}
\usepackage{color}
\usepackage[normalem]{ulem}
\usepackage{float}
\usepackage{amsthm}

\usepackage{ulem}

\usepackage{hyperref}
\definecolor{purple}{rgb}{0.75,0,0.25}
\hypersetup{
  colorlinks   = true, 
  urlcolor     = blue, 
  linkcolor    = blue, 
  citecolor   = purple 
}

\usepackage{color}

\usepackage[dvipsnames]{xcolor}

\usepackage[titletoc,title]{appendix}

\allowdisplaybreaks

\input epsf.tex

\newcommand{\ds}{\displaystyle}

\newtheorem{theorem}{Theorem}

\newtheorem{corollary}{Corollary}

\newtheorem{remark}{Remark}

\def\vu{\boldsymbol{u}}

\def\a{\alpha}
\def\b{\beta}
\def\r{\rho}

\def\T{\tau}

\def\g{\nabla}
\def\p{\partial}

\def\dint{\displaystyle\int}
\def\dsum{\displaystyle\sum}

\def\dfrac{\displaystyle\frac}

\def\dqh{\delta_h}
\def\dqq{\delta_q}

\def\bd{\boldsymbol}

\def\wtd{\widetilde}

\def\mb{\mbox}

\def\F{\mathcal{F}}

\definecolor{green3}{rgb}{0,0.8,0}
\title{Shallow Water Moment models for bedload transport problems}
\author{J. Garres-D{\'i}az \thanks{Dpto. Matem{\'a}ticas. Edificio Einstein - Universidad de C\'ordoba, Spain, (jgarres@uco.es, Tomas.Morales@uco.es)},
M.~J. Castro D\'iaz\thanks{Dpto. An\'alisis Matem\'atico. Universidad de M\'alaga, 14071 M\'alaga, Spain, (castro@anamat.cie.uma.es)},
J. Koellermeier\thanks{Department of Computer Science, KU Leuven, Celestijnenlaan 200a, 3001 Leuven, Belgium, (julian.koellermeier@kuleuven.be)}, T. Morales de Luna$^*$ }

\begin{document}
\maketitle
\date{}

\abstract{
	In this work a simple but accurate shallow model for bedload sediment transport is proposed. The model is based on applying the moment approach to the Shallow Water Exner model, making it possible to recover the vertical structure of the flow. This approach allows us to obtain a better approximation of the fluid velocity close to the bottom, which is the relevant velocity for the sediment transport. A general Shallow Water Exner moment model allowing for polynomial velocity profiles of arbitrary order is obtained. A regularization ensures hyperbolicity and easy computation of the eigenvalues. The system is solved by means of an adapted IFCP scheme proposed here. The improvement of this IFCP type scheme is based on the approximation of the eigenvalue associated to the sediment transport.  Numerical tests are presented which deal with large and short time scales. The proposed model allows to obtain the vertical structure of the fluid, which results in a better description on the bedload transport of the sediment layer. 
}

\vspace*{0.3cm}


\vspace*{0.4cm}

\bigskip



\frenchspacing   
\section{Introduction}

Sediment transport  and the morphological evolution of riverbeds due to the deposition and erosion are an active topic in the study of  fluvial processes. 
The sediment is transported  by the river current as suspended load (finer fractions carried by the flow) and  bedload (coarse fractions which move close to the bottom rolling, jumping and sliding), see \cite{librosed}. 

The study of sediment transport  focuses on understanding the relationship that exists between the movement of water and the movement of sedimentary materials. We are interested here in the so-called bedload transport, which is the type of transport that mainly happens near the bottom. In bedload, the sediment grains roll or slide along the bed. Single grains may even jump over the bed a length proportional to their diameter, losing for instants the contact with the soil, but mainly staying near the bed.  A first approach to model 
bedload transport is to couple the Shallow Water equations with the so-called Exner equation (see \cite{exner:1925}). This equation depends on the empirical definition of the solid transport flux for the bedload transport. Several formulations have been proposed, see for instance \cite{meyerPeter:1948,grass:1981,vanrijn}. This approach has 
been extensively used to describe bedload transport, see 
\cite{caleffi,castrodiaz2008sediment,liu3,canestrelli,murillo2010,juez}, among many others.

The description of these empirical formulae for the solid transport discharge is usually based on the velocity of the fluid, which is given by the Shallow Water model describing the hydrodynamic component. Nevertheless, this velocity corresponds to the averaged value in the water column. One would expect to use the near bed velocity of the fluid for bedload transport (see \cite{fernandez-nieto2017formal} and references there in), however the lack of the vertical profile for the velocity due to the Shallow Water approach makes it impossible. 

In recent years, effort has been made in using more complex shallow type models in order to have a better description of the fluid in the vertical direction. One possible direction is the use of the multilayer approach  \cite{audusse2010multilayer, fernandez-nieto2013multilayer, moralesdeluna2017derivation}. This approach allows us to recover the vertical profile by subdividing the domain along the normal direction in shallow layers, and applying the thin-layer hypothesis within each layer. Thus, the domain is discretized in the vertical direction, leading to a system with $N+1$ equations and unknowns, where $N$ is the number of layers. A drawback of this approach is that many layers should be employed if very complex profiles of velocity have to be recovered, leading to a high computational cost (although much lower than solving the full 3D Navier-Stokes system). In \cite{bonaventura:2018}, an application to bedload transport problem is simulated by using a multilayer model and the Grass formula, which allows the authors to use the velocity near the bottom and not the depth-averaged velocity as in the Shallow Water model. It should be noted that the resulting multilayer model seems to be hyperbolic based on numerical simulations, although this question remains open for arbitrary numbers of layers. Moreover, no analytical explicit general expression for the eigenvalues is known.

An extended Shallow Water model was derived in \cite{kowalski:2018}. The model uses a polynomial expansion of the horizontal velocity along the vertical axis such that complicated velocity profiles can be represented efficiently using an extended set of variables that includes the basis coefficients of the polynomial basis. This approach is called moment method and the resulting model is the Shallow Water Moment model (SWM). Even though the model showed good results for standard test cases, it was shown in \cite{koellermeier:2019} that the hyperbolicity of the SWM model is limited to a bounded domain in phase space. This is a known deficiency of standard moment models, for example, in kinetic theory \cite{Cai:2013,Koellermeier:2014,Koellermeier:2017}. The lack of hyperbolicity under certain flow conditions can lead to instabilities in numerical simulations. This drawback was solved in \cite{koellermeier:2019} by means of a hyperbolic regularization of the SWM model. The resulting Hyperbolic Shallow Water Moment model (HSWM) was proven to be hyperbolic under any flow conditions for arbitrary velocity profiles. Furthermore, the eigenvalues of the HSWM system can be computed analytically, which makes the application of numerical schemes easier. The HSWM model is thus an ideal starting point for the development of an improved sediment transport model that takes into account more complicated vertical velocity profiles.\\

The goal of this work is to propose such a model for bedload sediment transport problems. The model is obtained by applying the moment approach to extend the classical Shallow Water Exner model. After some analysis, the final model is obtained, which is expected to be hyperbolic, at least in the regime where the Shallow Water Exner  model is also hyperbolic. Actually, approximating the eigenvalues of the proposed model is similar to approximating the eigenvalues of the Shallow Water Exner model and the HWSM model, since almost all eigenvalues can be explicitly computed as we rigorously prove. This approach makes it possible to approximate the velocity at the bottom, that is the one used to sweep the sediment, thus improving the sediment transport.  

The paper is organized as follows: Section \ref{se:model} is devoted to presenting the model, where a brief review of the moment approach for shallow flows is also given. In Section \ref{se:num_aprox} the numerical scheme is presented, showing in particular the method to approximate the eigenvalues of the system. Some academic numerical tests are shown in Section \ref{se:tests}, together with a comparison with laboratory experiments. Finally, some conclusions are given in Section \ref{se:conclusions}.

\section{Shallow Water Moment models}\label{se:model}
In this section the initial system and the reference system, as well as the moment approximation framework is introduced to obtain the final Hyperbolic Shallow Water Exner Moment system (denoted HSWEM hereinafter) that we use to simulate bedload sediment transport.
\subsection{Incompressible Navier-Stokes and Exner sediment transport}
 In this work the model will be derived, for the sake of simplicity, in the 2D case, but it can be easily extended to the three dimensional case. Firstly, we consider a $2$D-Cartesian reference system where $(x,z)$ are the horizontal and vertical direction, respectively, with 
 $\vu = (u,w)$ the velocity vector. Then, we start from the 2D incompressible Navier-Stokes system. For a fluid with constant density $\r$,  the system reads
\begin{subequations}\label{eq:NS_ini}
\begin{equation}\label{eq:NS}
\left\{
\begin{array}{l}
\p_{x} u+ \p_{z}w= 0,\\[3mm]
\r\big(\p_{t}u+u\,\p_{x}u +w\,\p_{z}u\big)+\p_{x}p =\p_{x}\sigma_{xx}+\p_{z}\sigma_{xz},\\[3mm]
\r\big(\p_{t}w+u\,\p_{x}w +w\,\p_{z}w\big)+\p_{z}p =-\r g+\p_{x}\sigma_{zx}+\p_{z}\sigma_{zz},
\end{array}
\right.
\end{equation}
with $g$ the gravitational acceleration, $\sigma = \mu \left(\g\vu + \g\vu^T\right)$ the deviatoric tensor (whose components are denoted $\sigma_{xx}, \sigma_{xz},\sigma_{z,x},\sigma_{zz}$), $\mu$ the viscosity coefficient, and $p$ the pressure. 

Following a similar dimensional analysis as carried out in \cite{kowalski:2018} (see also the classical asymptotic analysis for the Shallow Water system \cite{gerbeau:2001}), the only viscous effect retained is the term $\partial_z \sigma_{xz}$. Moreover, a hydrostatic pressure is obtained, which means that the third equation in (\ref{eq:NS}) reduces to
\begin{equation}
p(t,x,z) = \r g \left(b(t,x)+h(t,x) -z\right),
\end{equation}
\end{subequations}
where $h(t,x)$ is the water height and $b(t,x)$ is the bottom, which may evolve in time. Thus, the free surface level is given by $b+h$.

In this paper, we shall assume that the bottom topography evolves due to the interaction between the fluid and the sediment particles that constitute the bottom. More explicitly we consider the bedload transport of these particles by means of the Exner equation \cite{exner:1925} 
\begin{equation}
\label{eq:exner}
\p_t b + \p_x Q_b = 0,
\end{equation}
where $Q_b$ is the solid transport discharge. Note that here we do not consider sediment transport in suspension and we do not include erosion-deposition effects. Only bedload transport is assumed here for the sake of simplicity, since the goal is to improve the vertical description of the velocity, leading to an improvement of the bedload sediment transport. However, erosion and deposition effects may be relevant in some situations as it was shown in \cite{gonzalezAguirre:2020}. Many empirical formulae may be found for erosion and deposition fluxes, which depend in different parameters to be calibrated. This is not the purpose of this paper and it will be studied in future works. 
To close the system, a definition of $Q_b$ must be given, which is usually defined by means of an empirical formula: Grass, Meyer-Peter \& M\"uller, Ashida \& Michiue,... see for instance \cite{meyerPeter:1948,ashida:1972, grass:1981,fernandezNieto:2016}. In this work we consider the Meyer-Peter \& M\"uller formula, which defines the solid transport discharge, in nondimensional form, as
\begin{subequations}
	\label{eq:MPM_full}
\begin{equation}
\label{eq:meyerPeterMuller_a}
\dfrac{Q_b}{Q} = sgn(\tau)\dfrac{8}{1-\varphi} \left(\theta -\theta_c\right)_+^{3/2},
\end{equation}
where $Q = d_s\sqrt{g\left(1/r-1\right)d_s}$ is the characteristic discharge, $(\cdot)_+$ is the positive part and $sgn(\cdot)$ is the sign function. In addition, $\T$ is the shear stress at the bottom, $\varphi$ is the porosity, $r=\r/\r_s$ with $\r_s$ the sediment density and $d_s$ the diameter of the sediment particles. The Shields parameter $\theta$ is defined as 
\begin{equation}\label{eq:meyerPeterMuller_b}
\theta = \dfrac{|\T|d_s^2}{g\left(\r_s-\r\right)d_s^3},
\end{equation}
and $\theta_c$ is the constant critical Shields stress. Note that, without loss of generality, any other formula for $Q_b$ may be used. Finally, for the shear stress at the bottom, the Manning friction law is assumed, leading to 
\begin{equation}\label{eq:meyerPeterMuller_c}
\T = \r g h S_f \quad\mb{with}\quad S_f = \dfrac{n^2 |u_{|_{z=b}}|}{h^{4/3}}u_{|_{z=b}},
\end{equation}
where $n$ is the Manning coefficient.
\end{subequations}

Regarding the boundary conditions, the usual kinematic conditions at the free surface and the bottom are used: 
$$
\p_t(b+h) + u_{|_{z=b+h}}\p_x(b+h) - w_{|_{z=b+h}} = 0,
$$
and 
$$
\p_tb + u_{|_{z=b}}\p_xb - w_{|_{z=b}} = 0.
$$
In addition, it is assumed that the deviatoric tensor vanishes on the free surface and it reduces to the usual friction condition at the bottom, which are written as
$$
\sigma_{xz|_{z=b+h}} =0\quad\mbox{and}\quad\sigma_{xz|_{z=b}} = \T,
$$ with $\T$ defined by \eqref{eq:meyerPeterMuller_c}. In the next subsection, we recall the classical shallow model including the Exner equation.

\subsection{Shallow Water Exner model}
Starting from the incompressible Navier-Stokes equations, the classical shallow water equations can be derived by averaging over the vertical variable assuming constant velocity in the horizontal direction. This is a well-known procedure and leads to the following model for water height $h$ and constant velocity $u_m$:
\begin{equation*}\label{eq:SW}
\left\{
\begin{array}{l}
\partial_t h + \partial_x( h u_m) = 0, \\[3mm]
\partial_t(h u_m) + \partial_x \left(h u_m^2 + \frac{1}{2}g h^2\right) = \frac{1}{\rho} \partial_z \sigma_{xz|_{z=b}} ,
\end{array}
\right.
\end{equation*}
where $\sigma_{xz|_{z=b}}$ is the friction from the contact with the bottom. 

Using the Exner model, and the Manning law for the friction term, the combined Shallow Water Exner (SWE) model reads
\begin{equation*}\label{eq:SWE}
\left\{
\begin{array}{l}
\partial_t h + \partial_x( h u_m) = 0, \\[2mm]
\partial_t(h u_m) + \partial_x \left(h u_m^2 + \frac{1}{2}g h^2\right)  = - \dfrac{g n^2|u_{m}|}{h^{1/3}}u_m, \\[4mm]
\p_t b + \p_x Q_b = 0,
\end{array}
\right.
\end{equation*}
or in conservative variables $\bd{W} = \left(h,hu_m,b\right)^t$, 
\begin{equation}
    \p_t \bd{W} + \bd{A}_{SWE}(\bd{W})\p_x \bd{W} = \bd{E}(\bd{W}),
    \label{e:SWE}
\end{equation}
with transport matrix
\begin{equation}
    \bd{A}_{SWE}(\bd{W}) = \left(\begin{matrix}
    0 & 1 & 0 \\[3mm]
    -u_m^2+gh & 2u & gh\\[3mm]
    \dqh & \dqq & 0
    \end{matrix}\right),
\end{equation}
where the notation $\dqh = \p_{h} Q_b$ and $\dqq = \p_{hu_m} Q_b$ is used. Remark that the Meyer-Peter $\&$ M\"uller formula for the solid transport discharge leads to 
\begin{subequations}
	\label{eq:dqb}
\begin{equation}
\dqq = \p_{hu_m}Q_b = \p_\theta Q_b\,\p_{hu_m}\theta = \dfrac{24 n^2}{\left(1-\varphi\right)\left(1/r-1\right)d_s} \left(\theta-\theta_c\right)^{1/2}_+ \dfrac{u_b}{h^{4/3}},
\end{equation}
and 
\begin{equation}
\dqh = \p_{h}Q_b =-\dfrac76 u_b \dqq.
\end{equation}
\end{subequations}

The right-hand side friction term is consistently written as $\bd{E}(\bd{W}) =  \left(0,- \dfrac{g n^2|u_{m}|}{h^{1/3}}u_m,0\right)^t$.\\

\medskip

Considering the physical properties of the system, it is possible to compute an equation for the propagation speeds via the characteristic polynomial $P_{\bd{A}_{SWE}}$ of the system matrix $\bd{A}_{SWE}(\bd{W})$ 
\begin{equation}\label{e:SWE_ev}
    P_{\bd{A}_{SWE}}(\lambda) = -\lambda \left((\lambda-u_m)^2 -gh \right) + gh (\delta_h +\lambda \delta_q).
\end{equation}

For the SWE system in \eqref{e:SWE}, it is not possible to have an easy explicit expression for the propagation speeds $\lambda,\, s.t.\, P_{\bd{A}_{SWE}}(\lambda) = 0$, but the Cardano's formula could be used to compute them. Nevertheless, as it is shown in \cite{cordier:2011}, the hyperbolicity of the model may be studied in an easy way. For the particular case of Manning's friction law used here, it is shown that the system is hyperbolic for the case of realistic and physical values ($|u_m|< 6 \sqrt{gh}$), although one could find complex eigenvalues otherwise. 
Note that in case of $\dqh=\dqq=0$, i.e. no sediment transport, the system has propagation speeds $\lambda_{1,2} = u_m \pm \sqrt{gh}$, $\lambda_3 = 0$, resembling standard shallow water flow with resting bottom topography. Whenever $\dqh, \dqq$ are sufficiently small, one expects to have three eigenvalues that are close to the ones obtained for the case of shallow water as the characteristic polynomial \eqref{e:SWE_ev} changes continuously with its coefficients.  \\

In the next subsection, we develop the moment approximation to system \eqref{eq:NS_ini}-\eqref{eq:MPM_full}, which is one of the contributions of this work.

\subsection{Shallow Water Exner Moment model}
The Shallow Water Exner model is a significant simplification as it assumes constant horizontal velocity along the vertical $z-$direction, but is often used to compute simple solutions of applications. The main drawback of the standard Shallow Water equations is that vertical variations on the velocity cannot be represented. This is especially important for sediment transport problems, where the bottom velocity leads to the crucial friction between the fluid and sediment. 
In \cite{kowalski:2018}, a new model for shallow flows that allows for changes in the vertical structure of the velocity is derived. Let us remind the reader briefly of this moment approximation technique for shallow flows.\\

\medskip

Firstly, a new variable, based on the $\sigma$-coordinates, is considered,
\begin{equation}
\label{eq:sigmacoord}
\xi = \dfrac{z-b}{h},\quad\mb{where }\xi\in [0,1]\mb{ for } z\in [b,b+h],
\end{equation}
with $\xi = 0$ and $\xi = 1$ corresponding to the bottom $z=b$ and the free surface $z=b+h$, respectively. Note that, denoting by $\wtd{\psi}(t,x,\xi) = \psi(t,x,\xi h +b)$, the differential operators read
$$
\p_\xi\wtd{\psi} = h\p_z\psi
$$
and 
$$
h\p_s\psi = \p_s\left(h\wtd{\psi}\right) - \p_\xi \left(\p_s\left(\xi h + b\right)\wtd{\psi}\right),\quad\mbox{ for } s\in\{t,x\}.
$$
Taking into account this mapping, the Navier-Stokes-Exner system \eqref{eq:NS_ini}-\eqref{eq:MPM_full} is rewritten as
\begin{equation}\label{eq:NS_sigma}
\left\{
\begin{array}{l}
\p_{x}\left(h \wtd{u}\right)+ \p_{\xi}\left(\wtd{w}-\wtd{u}\p_x\left(\xi h +b\right)\right)= 0,\\[3mm]
\p_{t}\left(h\wtd{u}\right)\,+\,\p_{x}\left(h \wtd{u^2}\right) \,+\, gh\p_x\left(b+h\right) + \p_\xi \left(h\omega \wtd{u}\right) = \dfrac{1}{\r}\p_{\xi}\sigma_{xz},\\[3mm]
\p_t b + \p_x Q_b = 0,
\end{array}
\right.
\end{equation}
where a vertical coupling term $\omega$ is defined as
$$
h\omega = -\xi\p_t h - \p_x \left(h\dint_{0}^\xi \wtd{u}\,d\xi\right)
$$
and the boundary conditions are easily written in the variable $\xi$, obtaining
$$
\begin{array}{l}
\p_tb + \wtd{u}_{|_{\xi=0}}\p_xb - \wtd{w}_{|_{\xi=0}} = 0,\\[3mm]
\p_t(b+h) + \wtd{u}_{|_{\xi=1}}\p_x(b+h) - \wtd{w}_{|_{\xi=1}} = 0,\end{array}\qquad\mbox{and}\qquad \begin{array}{l}
\wtd{\sigma}_{xz|_{\xi=0}} = \T,\\[3mm]
\wtd{\sigma}_{xz|_{\xi=1}} =0.
\end{array}
$$

The main idea of the moment approximation is to consider an expansion of the velocity in the vertical variable $\xi$ as
\begin{equation}\label{eq:u_moment}
\wtd{u}(t,x,\xi) = u_m (t,x) \,+\, \dsum_{j=1}^N \alpha_{j}(t,x)\phi_j\left(\xi\right),
\end{equation}
where $u_m$ is the mean velocity, which does not depend on the vertical direction, $\phi_j:[0,1]\longrightarrow \mathbb{R}$ is the scaled Legendre polynomial of degree $j$, and $\alpha_j$, called hereinafter \textit{moment}, is the corresponding basis coefficient. 

In general, $N \in \mathbb{N}$ can be arbitrary. The first basis functions are given by
$$
\phi_0(\xi) = 1,\quad  \phi_1(\xi) = 1-2\xi,\quad \phi_2(\xi) = 1-6\xi+6\xi^2,\quad \phi_3(\xi) = 1-12\xi+30\xi^2-20\xi^3.
$$
Note that these polynomials fulfill $\phi_i(0) = 1$ and $\dint_{0}^1 \phi_i(\xi)\,d\xi = 0$, for $i\ne0$. Then, effectively, the mean velocity is $u_m = \dint_{0}^1 \wtd{u}\,d\xi$. 

Finally, the second equation in \eqref{eq:NS_sigma} is tested (i.e. multiplied) with the same ansatz functions $\phi_j$ as the basis function to get
$$\phi_j \p_{t}\left(h\wtd{u}\right)\,+\,\phi_j\p_{x}\left(h \wtd{u^2}\right) \,+\, \phi_j gh\p_x\left(b+h\right) + \phi_j\p_\xi \left(h\omega \wtd{u}\right) = \phi_j\dfrac{1}{\r}\p_{\xi}\sigma_{xz}, \quad j=0,1,2,...,N,
$$
which are integrated for $\xi\in[0,1]$, to obtain the final Shallow Water Exner Moment (SWEM) system.\\

\medskip
 The only term that is different from previous models is the friction term, as a standard Newton fluid with slip boundary condition at the bottom was used in \cite{kowalski:2018}. The resulting model for the Exner friction term is obtained after the following calculations.

On the one hand, for $j=0$, i.e. the momentum equation, we have that
$$
\dint_{0}^1 \p_{\xi} \wtd{\sigma}_{xz}d\xi = \wtd{\sigma}_{xz|_{\xi=1}}-\wtd{\sigma}_{xz|_{\xi=0}}= - \T = - \dfrac{\r g n^2|u_b|}{h^{1/3}}u_b,
$$
where $u_b = \wtd{u}_{|_\xi=0} = u_m + \dsum_{j=1}^N \a_j$ denotes the velocity at the bottom. On the other hand, for $j\neq 0$, i.e. the additional moment equations, we obtain
$$
\begin{array}{ll}\dint_{0}^1 \phi_i \p_{\xi} \wtd{\sigma}_{xz}d\xi =& \dint_{0}^1 \p_{\xi} \left(\phi_i\wtd{\sigma}_{xz}\right)d\xi -\dint_{0}^1\wtd{\sigma}_{xz} \p_{\xi} \phi_i d\xi =  - \dfrac{\r g n^2|u_b|}{h^{1/3}}u_b - \dfrac{\mu}{h}\dint_{0}^{1} \p_{\xi}\wtd{u}\p_{\xi}\phi_i d\xi\\[4mm]
&= - \dfrac{\r g n^2|u_b|}{h^{1/3}}u_b - \dfrac{\mu}{h}\dsum_{j=1}^N\a_jC_{ij},
\end{array}
$$
where $C_{ij} = \dint_{0}^{1} \p_{\xi}\phi_i\p_{\xi}\phi_j d\xi$. 

Following the derivations above, the general SWEM system with $N+3$ equations and unknowns reads:
\begin{equation}\label{eq:model_SWM}
\left\{
\begin{array}{l}
\p_t h + \p_{x}\left(h u_m\right)= 0,\\[3mm]
\p_{t}\left(h u_m\right)\,+\,\p_{x}\left(h u_m^2 + h\dsum_{j=1}^N \dfrac{\a_j^2}{2j+1}\right) \,+\, gh\p_x\left(b+h\right) = - \dfrac{g n^2|u_{b}|}{h^{1/3}}u_b,\\[3mm]
\p_{t}\left(h \a_i\right)\,+\,\p_{x}\left(h \left(2u_m\a_i + \dsum_{j,k=1}^N A_{ijk} \a_j\a_k\right)\right) = u_m\p_x\left(h\a_i\right) - \dsum_{j,k=1}^N B_{ijk}\a_k \p_x\left(h\a_j\right)\\[3mm]
\qquad - (2i+1)\left(\dfrac{g n^2|u_{b}|}{h^{1/3}}u_b + \dfrac{\nu}{h}\dsum_{j=1}^N C_{ij}\a_j\right),\\[3mm]
\p_t b + \p_x Q_b = 0,
\end{array}
\right.
\end{equation}
where $Q_b$ is defined by \eqref{eq:MPM_full}, and we recall that the bottom velocity is $u_b = u_m + \dsum_{j=1}^N \a_j$. The constant $\nu = \mu/\r$ is the kinematic viscosity, and $A_{ijk},B_{ijk},C_{ij}$ are constant coefficients depending on the Legendre polynomials. Concretely,
$$
\dfrac{A_{ijk}}{2i+1} = \dint_{0}^1 \phi_i\phi_j\phi_k d\xi,\quad \dfrac{B_{ijk}}{2i+1} = \dint_{0}^1 \phi_i' \left(\dint_{0}^\xi \phi_jd\xi\right) \phi_k d\xi,\quad\mbox{and}\quad \dfrac{C_{ij}}{2i+1} = \dint_{0}^1 \phi_i'\phi_j' d\xi.
$$

The previous system can be written in form of a first-order system with non-conservative products and source terms, which are resulting from the friction term. Thus, system \eqref{eq:model_SWM} reads
\begin{equation}
\label{eq:SWEM}
\p_t \bd{W} + \bd{F}(\bd{W}) = \bd{B}(\bd{W})\p_x \bd{W} + \bd{E}(\bd{W}),
\end{equation}
where $\bd{W} = \left(h,hu_m,h\a_1,\dots,h\a_N,b\right)^t$ is the vector of conservative variables, the convective flux is $\bd{F}(\bd{W})=\left(F_1(\bd{W}),F_2(\bd{W}),\dots,F_{N+3}(\bd{W})\right)^t$, with $$\bd{F}_i(\bd{W}) = \left\{\begin{array}{lrl}
hu_m &\quad \mbox{if}&i = 1,\\[3mm]
h u_m^2 \,+\, g\dfrac{h^2}{2}\,+\, h\dsum_{j=1}^N \dfrac{\a_j^2}{2j+1} & \mbox{if}&i = 2,\\[3mm]
h \left(2u_m\a_i \,+\,\dsum_{j,k=1}^N A_{ijk} \a_j\a_k\right) & \mbox{if}&i \in \{3,...,N+2\},\\[3mm]
Q_b & \mbox{if}&i =N+3,\\[3mm]
\end{array}\right.$$
$\bd{B}(\bd{W})$ the non-conservative terms, that will be detailed later, and the source term
$$\bd{E}(\bd{W})= \left\{\begin{array}{lrl}
0 &\quad \mbox{if}&i = 1,N+3,\\[3mm]
-\dfrac{g n^2|u_{b}|}{h^{1/3}}u_b & \mbox{if}&i = 2,\\[3mm]
- (2i+1)\left(\dfrac{g n^2|u_{b}|}{h^{1/3}}u_b + \dfrac{\nu}{h}\dsum_{j=1}^N C_{ij}\a_j\right) & \mbox{if}&i \in \{3,...,N+2\}.\\[3mm]
\end{array}\right.$$

Similarly as for the Shallow Water Moment model without sediment transport in \cite{kowalski:2018}, an investigation of the propagation speeds is very difficult for the model, due to the strong nonlinearity of the system.

\subsection{Hyperbolic Shallow Water Exner Moment model}
For the Shallow Water Moment model without sediment transport in \cite{kowalski:2018}, it has been shown in \cite{koellermeier:2019} that the model is not hyperbolic for for $N>1$. Depending on the flow variables $\bd{W}$, the eigenvalues of the system could become imaginary, which leads to instabilities and oscillations in numerical simulations. However, the case $N>1$ is especially interesting in the case of sediment transport as it allows for a more complex structure of the flow and changes the bottom velocity, such that a more accurate description of the sediment transport can be expected.

The SWEM model is thus not usable for our simulations. However, in \cite{koellermeier:2019} a modified model called Hyperbolic Shallow Water Moment model (HSWM) is introduced to ensure hyperbolicity independent of the value of $\bd{W}$ for the case without sediment transport. We will use the same strategy here to achieve hyperbolicity of the moment part of the system and to allow for an analysis of the propagation speeds. To apply the hyperbolic regularization used in \cite{koellermeier:2019}, we start from the transport part of the previous system in the form of a first-order system with non-conservative products
$$
\p_t \bd{W} + \bd{F}(\bd{W}) = \bd{B}(\bd{W})\p_x \bd{W},
$$ and equivalently rewrite it as 
$$
\p_t \bd{W} + \bd{A}(\bd{W})\p_x \bd{W} = 0
$$
with $\bd{A}=\dfrac{\partial\bd{F}}{\partial\bd{W}} - \bd{B}$. It was already mentioned in \cite{kowalski:2018} that the system matrix $\bd{A}(\bd{W})$ can have imaginary eigenvalues. 
The hyperbolic HSWM model is now deduced by modifying the system matrix $\bd{A}(\bd{W})$ such that the sub-matrix including the first $N+2$ rows and $N+2$ columns has real eigenvalues, i.e., the part concerning the fluid transport and not the sediment transport. We denote the resulting matrix as $\bd{A}_H(\bd{W})$, which is obtained from $\bd{A}(\bd{W})$ by keeping the terms in $\a_1$ and forcing all the high-order moment terms of the submatrix to be zero, i.e., $\a_2=\dots=\a_N = 0$ in the submatrix of the first $N+2$ rows and columns of the system matrix. Even though this seems like a significant change of the model, it was shown in \cite{koellermeier:2019} that the accuracy of the corresponding HSWM model is sufficient, similar to the same strategy used in kinetic theory, see \cite{Cai:2013,Koellermeier:2014,Koellermeier:2017}. The system matrix $\bd{A}_H(\bd{W})$ can be seen as a linearization of the original system matrix around the first-order deviation from equilibrium, which is justified by the fact that the coefficients $\a_2=\dots=\a_N$ are usually small. \\

Following the derivations in \cite{koellermeier:2019}, the system matrix $\bd{A}_H(\bd{W}) \in \mathbb{R}^{(N+3)\times(N+3)}$ of the Hyperbolic Shallow Water Exner Moment system can easily be obtained as
\begin{equation}
    \bd{A}_H(\bd{W})=\left(
    \begin{array}{cc|ccccc|c}
    & 1 &  &  &  &  &  &  \\
     gh-u_m^2-\frac{1}{3}\alpha_1^2& 2 u_m & \frac{2}{3}\alpha_1 &  &  &  & & gh \\ \hline
    -2u_m \alpha_1 & 2\alpha_1 & u_m & \frac{3}{5}\alpha_1  & & & \\
    -\frac{2}{3}\alpha_1^2 &  & \frac{1}{3}\alpha_1 & u_m & \ddots  & & \\
    &  &  & \ddots & \ddots  & \frac{N+1}{2N+1}\alpha_1 & \\
    &  &  &  & \frac{N-1}{2N-1}\alpha_1 & u_m &  &\\ \hline
    \dqh & \dqq &\dqq  & \hdots & \dqq &  \dqq &  &\\
    \end{array}
    \right),
    \label{e:HSWEM}
\end{equation}
where again $\dqh = \p_{h} Q_b$ and $\dqq = \p_{hu_m} Q_b$. Note that the bottom row contains the same values $\dqq$ in the different $N+1$ entries due to the simple evaluation of the bottom velocity. The hyperbolic regularization allows to write the system matrix explicitly. The system matrix does no longer contain an explicit dependence on the higher moments $\alpha_i$ for $i>1$ except for the terms $\dqh,\dqq$. However, the equations for the higher moments still couple with the other equations so that the system is still very non-linear.

Note that the system HSWEM can be written as 
\begin{subequations}\label{eq:HSWEM}
	\begin{equation}
\label{eq:HSWEM_A}
\p_t \bd{W} + \bd{F}(\bd{W}) = \wtd{\bd{B}}(\bd{W})\p_x \bd{W},
\end{equation}
 where \begin{equation}
 \label{eq:HSWEM_B}
 \wtd{\bd{B}}(\bd{W}) = \bd{B}(\bd{W}) - (\bd{A}_H(\bd{W})-\bd{A}(\bd{W})) = \dfrac{\partial\bd{F}}{\partial\bd{W}} - \bd{A}_H(\bd{W}).
 \end{equation}
 \end{subequations}
That is, the Hyperbolic Shallow Water Exner Moment system can be seen as the Shallow Water Exner Moment system with modified non-conservative terms.

For the sake of completeness, let us write explicitly the matrix $\dfrac{\partial\bd{F}}{\partial\bd{W}}\in \mathbb{R}^{(N+3)\times(N+3)}$
\begin{equation}\label{eq:jacobianoF}
\dfrac{\partial\bd{F}}{\partial\bd{W}} = \left(
\begin{array}{cc|ccc|c}
 & 1 &   & &  & \\[2mm]
gh -u_m^2-\sum_{j=1}^N \frac{\alpha_j^2}{2j+1} & 2u_m & \frac{2}{3} \alpha_1 &  \dots  & \frac{2}{2N+1}\alpha_N & \ \  \\[2mm] \hline
-2u_m\alpha_1-E_1  & 2\alpha_1 & 2u_m+D_{11}&  \dots & D_{1N}& \\[2mm]
\vdots &\vdots & \vdots & \ddots & \vdots &  \\[2mm]
-2u_m\alpha_N-E_N  & 2\alpha_N & D_{N1} &  \dots & 2u_m+D_{NN} & \\[2mm] \hline
\delta_h & \delta_q & \delta_q &   \ldots & \delta_q&
\end{array}
\right),
\end{equation}
with $E_i= \ds\sum_{j,k=1}^N A_{ijk} \alpha_j\alpha_k$ and $D_{ij}=2\ds\sum_{k=1}^NA_{ijk}\alpha_k$.

\newpage

Interestingly, we can get the following result for the propagation speeds of the HSWEM system, which are the eigenvalues of the above matrix $\bd{A}_H(\bd{W})$.

\begin{theorem}
    \label{TheoremEV}
    The HSWEM system matrix $\bd{A}_H(\bd{W}) \in \mathbb{R}^{(N+3)\times(N+3)}$ \eqref{e:HSWEM} has the following characteristic polynomial
    \begin{equation*}
        \chi_{A} (\lambda) = \left[(-\lambda)\left( \left(\lambda - u_m\right)^2 -gh -\alpha_1^2 \right) + gh (\dqh + \lambda \dqq + 2\alpha_1 \dqq) \right] \cdot \chi_{A_2} (\lambda - u_m),
    \end{equation*}
    with $\chi_{A}(\lambda)=\mbox{det}(\bd{A}_H(\bd{W})-\lambda\bd{I})$, and $\chi_{A_2}(\lambda-u_m)=\mbox{det}(A_2-(\lambda-u_m)\bd{I})$ where $A_2 \in \mathbb{R}^{N \times N}$ is defined as follows

    \begin{equation} \label{A_2}
        A_{2}=\left(
        \begin{array}{ccccc}
            & c_2   &       & \\
        a_2 &       & \ddots& \\
            & \ddots&       & c_N  \\
            &       & a_N   &  \\
        \end{array}
        \right),
    \end{equation}
    with values $c_{i}=\frac{i+1}{2i+1}\alpha_1$ and $a_i=\frac{i-1}{2i-1}\alpha_1$ the values above and below the diagonal, respectively, from \eqref{e:HSWEM}.
\end{theorem}

\begin{proof}
    See Appendix \ref{app:charPoly}.
\end{proof}

\begin{remark} \label{re:eigenMoments}
    Theorem \ref{TheoremEV} states that the eigenvalue structure of the moment part within the HSWEM is preserved. The last $N$ propagation speeds are thus given by 
    $$
        \lambda_i = u_m + b_i \alpha_1, \qquad\mbox{for }i=4,\ldots,N+3,
    $$
    where $b_i \alpha_1$ are the real-valued eigenvalues of the matrix $A_2$ \cite{koellermeier:2019}, which can be computed explicitly.
\end{remark}
\begin{remark}
    The first three eigenvalues $\lambda_1,\lambda_2,\lambda_3$ are the roots of the polynomial 
    $$
        P_{A_H}(\lambda)=(-\lambda)\left( \left(\lambda - u_m\right)^2 -gh -\alpha_1^2 \right) + gh(\dqh + \lambda \dqq + 2\alpha_1 \dqq),
    $$
    which is a consistent extension of the standard Shallow Water Exner model, see Equation \eqref{e:SWE_ev}.
\end{remark}
\begin{remark}
    For $\alpha_1=0$, the same eigenvalues as for the SWE model will be obtained. As the roots of a polynomial change continuously with its coefficients, we will obtain similar eigenvalues $\lambda_1,\lambda_2,\lambda_3$ for small values of $\alpha_1$. It is clear that the main difference between the propagation speeds of both models lies in the value of $\alpha_1$. We thus expect differences in numerical simulations whenever $\alpha_1$ is reasonably large.
\end{remark}
\begin{remark}
    An analogous result to Theorem \ref{TheoremEV} can be derived if the model is not based on the HSWM, but on the slightly modified $\beta$-HSWM, derived in \cite{koellermeier:2019}. The $\beta$-HSWM differs from the HSWM in only one entry of the system matrix while it achieves similar accuracy as the HSWM model for standard shallow flows. The $\beta$-HSWM model has the benefit that the values $b_i$ of the corresponding matrix similar to $A_2$ are guaranteed to fulfill $b_i \in \left[ -1,1 \right]$. This makes it easier to obtain estimates for the maximum and minimum value of the propagation speeds within the numerical scheme. In this paper, we will only consider the HSWEM for conciseness and not use the $\beta$-HSWM model for simulations.
\end{remark}

In next section we detail the numerical approximation of the new HSWEM model.

\section{Numerical approximation}\label{se:num_aprox}
In the numerical approximation of the proposed model, one finds two different scenarios depending on the time-scale of the sediment dynamics. On the one hand, when very slow processes happen, there is a weak interaction between the sediment and the hydrodynamic counterpart. In this case, for which the computational time is very long, a decoupled discretization may be justified. This approach is for instance used in \cite{cunge_practical_1980,de_vriend_medium-term_1993,kubatko:2006}. On the other hand, we have the case of rapid movements, for which there is a strong interaction between fluid and sediment. In this case, the equation of the sediment must be coupled to the hydrodynamic counterpart, leading to a coupled discretization. Summarizing, the time-splitting for the sediment transport is only valid for those cases where the time-scale of the  morphodynamic problem is much slower than the time-scale of the hydrodynamic one, as discussed in \cite{cordier:2011}. Otherwise, a uncoupled treatment of the Exner equation would lead to stability problems. Here, a coupled discretization is considered that will be robust in any scenario, avoiding instabilities due to decoupling. 

In any case, it is essential to give a good approximation of the eigenvalues for different reasons. In the case of strong interactions, the approximation of the eigenvalue associated to the sediment is important for stability reasons. We refer the reader to \cite{cordier:2011}, where it is shown that the use of approximations given by shallow water eigenvalues results in wrong speed information for supercritical regimes in the Exner system. In the case of weak interactions, including some information on the intermediate eigenvalue associated to the bedload transport is essential, otherwise one would add too much numerical diffusion into the scheme for the sediment. In this case (long time simulations with a weak interaction), it is also important to consider high-order schemes to decrease the numerical diffusion.

We consider a finite volume method, based on a two-step approach: first, the friction terms in the fluid are neglected, which will then be considered in the second step by a semi-implicit approach. Remark that this is common approach for shallow water type systems (see \cite{garcia_simulacion_2000, castrodiaz2008sediment, mangeneycastelnau_numerical_2003} among others). When the high-order scheme is applied, we combine CWENO space reconstructions (see for example \cite{cravero2019optimal} and the references therein) with SI-RK3 introduced in \cite{chertock:15}. Let us consider that the horizontal domain is subdivided in control volumes $V_i = [x_{i-1/2},x_{i+1/2}]$, for $i\in\mathcal{I}$. For the sake of simplicity, a constant cell length $\Delta x$ will be considered. The center of the volume is $x_i = \left(x_{i-1/2}+x_{i+1/2}\right)/2$, and for a time $t^n$, the cell average is
$$
\bd{W}^n_i = \dfrac{1}{\Delta x}\dint^{x_{i+1/2}}_{x_{i-1/2}} \bd{W}(x,t^n) dx.
$$
Firstly, the hyperbolic system with non-conservative product is solved in the framework of path-conservative schemes \cite{pares:2004,pares:2006, castro2017well}. Secondly, a semi-implicit approach is considered to add friction terms.\\

\medskip

The hyperbolic system can be written as \eqref{eq:HSWEM}, and it is discretized as
\begin{subequations}\label{eq:scheme}
\begin{equation} \label{eq:step1}
\begin{array}{l}
\bd{W}_{i}^{n+1/2}=\bd{W}_i^n +\dfrac{\Delta t}{\Delta x} \left( \bd{\F}_{i-1/2}^n-\bd{\F}_{i+1/2}^n \,
+ \dfrac12\left(\bd{\mathcal{B}}_{i+1/2}^n+\bd{\mathcal{B}}_{i-1/2}^n \right)\right),
\end{array}
\end{equation}
with
\begin{equation}
\bd{\mathcal{B}}_{i+1/2} = \dint_{0}^1 \wtd{\bd{B}}\big(\Phi\left(s;\bd{W}_i,\bd{W}_{i+1}\right)\big)\dfrac{\p}{\p s}\Phi\left(s;\bd{W}_i,\bd{W}_{i+1}\right) ds, 
\end{equation}
where $\wtd{\bd{B}}$ is given by \eqref{eq:HSWEM_B}, and  $\Phi\left(s;\bd{W}_i,\bd{W}_{i+1}\right)$ is a path joining the two states. In this case, we consider as path the straight segments $\Phi\left(s;\bd{W}_i,\bd{W}_{i+1}\right) = \bd{W}_i + s\left(\bd{W}_{i+1}-\bd{W}_{i}\right)$. In practice, this integral is approximated by a quadrature rule:
$$
\bd{\mathcal{B}}_{i+1/2} = \dsum_{k=1}^{M} \omega_k \wtd{\bd{B}} \big(\Phi\left(s_k;\bd{W}_i,\bd{W}_{i+1}\right)\big)\left(\bd{W_{i+1}}-\bd{W_{i}}\right),
$$
where $\omega_k$ and $s_k$ are the weight and quadrature points, respectively, of the chosen quadrature formula. 
The numerical flux in \eqref{eq:step1} is written
\begin{equation}
\begin{array}{l}
\bd{\F}_{i+1/2} = \dfrac{1}{2} \left(\bd{F} (\bd{W}_i)+\bd{F}(\bd{W}_{i+1})\right) \,
- \,\dfrac{1}{2}\bd{\mathcal{D}}_{i+1/2}, \\
\end{array}
\end{equation}
\end{subequations}
where $\bd{\mathcal{D}}_{i+1/2}$ is the numerical diffusion of the scheme. In the previous expression we have dropped the time dependency for simplicity.

We consider here a method in the framework of Polynomial Viscosity Methods (PVM) introduced in \cite{castro:2012}, where the numerical diffusion  $\bd{\mathcal{D}}_{i+1/2}$ is defined in terms of a polynomial evaluation on the Roe matrix of the full non-conservative system. In particular, the IFCP method \cite{fernandezNieto:2011}, that is used in \cite{gonzalezAguirre:2020} to simulate sediment transport problems with erosion-deposition effects, will be used here as well. In \cite{gonzalezAguirre:2020}, the ideas introduced in \cite{cordier:2011} were used to approximate the eigenvalues of the system. This is crucial, namely in the case of a strong fluid-sediment interaction. This approach must be adapted and modified here in order to not fail in some particular configurations, as when $Fr \approx 1$. This will be discussed later and in the numerical tests section.\\

Let us focus now on the definition of the numerical scheme. The approximation of the eigenvalues will be discussed later. Let $\lambda_{m}<\lambda_{med}<\lambda_{p}$ be the approximations of the minimum, an intermediate, and the maximum eigenvalue. The numerical diffusion $\bd{\mathcal{D}}_{i+1/2}$ of the IFCP method could be written as follows (see \cite{fernandezNieto:2011}):

\begin{eqnarray}\label{eq:diffusion}
\bd{\mathcal{D}}_{i+1/2} &=& \b_1 \left(\bd{W}_{i+1}-\bd{W}_i\right) 
+ \b_2\left(\bd{F}(\bd{W}_{i+1})-\bd{F}(\bd{W}_{i}) + \bd{\mathcal{B}}_{i+1/2}\right) \\[3mm]
&&+ \b_3 \bd{\mathcal{A}}_{i+1/2}\left(\bd{F}(\bd{W}_{i+1})-\bd{F}(\bd{W}_{i}) + \bd{\mathcal{B}}_{i+1/2}\right)\nonumber,
\end{eqnarray}
where $\bd{\mathcal{A}}_{i+1/2}=\bd{A}_H(\bd{W}_{i+1/2})$, is the evaluation of the system matrix \eqref{e:HSWEM} on the intermediate (Roe) state $\bd{W_{i+1/2}}=\left(h_{i+1/2},(hu)_{i+1/2},(h\alpha_1)_{i+1/2},\dots,(h\alpha_N)_{i+1/2}\right)$, defined by 
$$
h_{i+1/2}=\dfrac{h_i+h_{i+1}}{2} \qquad\mbox{and}\qquad (h\varsigma)_{i+1/2} = h_{i+1/2}\dfrac{\varsigma_{i}\sqrt{h_i}+\varsigma_{i+1}\sqrt{h_{i+1}}}{\sqrt{h_i}+\sqrt{h_{i+1}}}, 
$$ for $\varsigma=u_m,\alpha_1,\dots,\alpha_N$.

The coefficients $\b_1,\b_2,\b_3$ in \eqref{eq:diffusion} are the solution of system 
$$
\left(\begin{matrix}
1 & \lambda_m & \lambda_m^2\\
1 & \lambda_{med} & \lambda_{med}^2\\
1 & \lambda_p & \lambda_p^2\\
\end{matrix}\right) \left(\begin{matrix}
\b_1\\
\b_2\\
\b_3
\end{matrix}\right) = \left(\begin{matrix}
|\lambda_m|\\
|\lambda_{med}|\\
|\lambda_p|\\
\end{matrix}\right).
$$
Therefore, defining $\gamma = \left(\lambda_{med}-\lambda_{m}\right)\left(\lambda_{p}\left(\lambda_{p}-\lambda_{m}-\lambda_{med}\right)+\lambda_{m}\lambda_{med}\right)$, we obtain
$$\b_1\gamma = |\lambda_m|\left(\lambda_{med}\lambda^2_{p}-\lambda^2_{med}\lambda_{p}\right) - |\lambda_{med}|\left(\lambda_{m}\lambda^2_{p}-\lambda^2_{m}\lambda_{p}\right) + |\lambda_p|\left(\lambda_{m}\lambda^2_{med}-\lambda^2_{m}\lambda_{med}\right),
$$
$$\b_2\gamma = -|\lambda_m|\left(\lambda^2_{p}-\lambda^2_{med}\right) + |\lambda_{med}|\left(\lambda^2_{p}-\lambda^2_{m}\right) - |\lambda_p|\left(\lambda^2_{med}-\lambda^2_{m}\right),
$$
$$
\b_3\gamma = |\lambda_m|\left(\lambda_{p}-\lambda_{med}\right) - |\lambda_{med}|\left(\lambda_{p}-\lambda_{m}\right) + |\lambda_p|\left(\lambda_{med}-\lambda_{m}\right).
$$

Unfortunately, the numerical scheme \eqref{eq:scheme} with the numerical diffusion \eqref{eq:diffusion} is not well-balanced for the water at rest stationary solution given by  
	$$
	u_m = \a_1 = \cdots = \a_N = 0,\qquad\mbox{and}\qquad b+h = constant.
	$$
In particular, the term $\b_1 \left(\bd{W}_{i+1}-\bd{W}_i\right)$ is the responsible of this fact. In order to obtain a well-balanced numerical scheme we propose to follow \cite{castro:2012} and replace  this term by 
$$
\b_1 \left(\bd{W}_{i+1}-\bd{W}_i + \bd{S}^*_{i+1/2} \right) \quad\mbox{where}\quad \bd{S}^*_{i+1/2}=\left(b_{i+1}-b_i,0,\dots,0\right)^t .
$$

Therefore,  the numerical diffusion reads as follows
\begin{eqnarray}\label{eq:diffusion2}
\bd{\mathcal{D}}_{i+1/2} &=& \b_1 \left(\bd{W}_{i+1}-\bd{W}_i -\bd{S^*}_{i+1/2} \right)+ \b_2\left(\bd{F}(\bd{W}_{i+1})-\bd{F}(\bd{W}_{i}) + \bd{\mathcal{B}}_{i+1/2}\right) \\[3mm]
&& + \b_3 \bd{\mathcal{A}}_{i+1/2}\left(\bd{F}(\bd{W}_{i+1})-\bd{F}(\bd{W}_{i}) + \bd{\mathcal{B}}_{i+1/2}\right)\nonumber.
\end{eqnarray}

\begin{remark}
	The numerical scheme above, with the numerical diffusion defined by \eqref{eq:diffusion2}, is well-balanced for steady state solutions corresponding to water at rest. More explicitly, the scheme preserves the solutions satisfying
	$$
	u_m = \a_1 = \cdots = \a_N = 0,\qquad\mbox{and}\qquad b+h = constant.
	$$

\end{remark}

\bigskip

Once the hyperbolic part of the system is solved, the friction terms are added using a semi-implicit discretization
$$\bd{W}_{i}^{n+1}=\bd{W}_i^{n+1/2} +\Delta t \bd{E}\left(\bd{W}_i^n,\bd{W}_{i}^{n+1}\right).$$

\noindent We trivially have $h^{n+1}_i = h_i^{n+1/2}$, and taking into account that $u_b = u_m + \dsum_{j=1}^N \a_j$, the equations
$$
\begin{array}{l}
h_i^{n+1}u_{m,i}^{n+1}=h_i^{n+1}u_{m,i}^{n+1/2} - \Delta t \dfrac{g n^2|u^n_{b,i}|}{h_i^{n,1/3}}u_{b,i}^{n+1},\\[3mm]
h^{n+1}_i \a_{k,i}^{n+1}=h_i^{n+1}\a_{k,i}^{n+1/2} - \Delta t (2k+1)\left(\dfrac{g n^2|u^n_{b,i}|}{h_i^{n,1/3}}u_{b,i}^{n+1} + \dfrac{\nu}{h_{i}^n}\dsum_{j=1}^N C_{kj}\a^{n+1}_{j,i}\right),\quad \mbox{for }k = 1,2,...,N,
\end{array}
$$
define a $N\times N$ linear system that can be exactly solved to find $u_{m,i}^{n+1},\a_{1,i}^{n+1},\a_{2,i}^{n+1},...,\a_{N,i}^{n+1}$, and therefore $\bd{W}_i^{n+1}$. 

\subsubsection*{Approximation of the eigenvalues}
As mentioned before, it is important to give an accurate approximation of the eigenvalues. In particular, an upper bound of the maximum absolute value of the eigenvalues is needed for the CFL condition. Moreover, some information on the wave speed associated to the sediment layer is needed, so that not too much numerical diffusion is added for the sediment movement. To this end, following Theorem \ref{TheoremEV}, we need to approximate the roots of the polynomial 
$P_{S}(\lambda) = f(\lambda) - d(\lambda)$, with
$$
f(\lambda) = \lambda\left(\left(u_m-\lambda\right)^2 - \left(gh+\a_1^2\right)\right),
$$
and
$$
d(\lambda) = gh\left(\lambda \dqq +\dqh+ 2\alpha_1 \dqq\right).
$$
Note that the rest of the eigenvalues can be explicitly computed (see Remark \ref{re:eigenMoments}). Note also
that $\dqq,\dqh$ are given by \eqref{eq:dqb}.\\

Now, the eigenvalues are the values $\lambda^*$ verifying $f(\lambda^*)=d(\lambda^*)$. To find these values, we use the approach in  \cite{gonzalezAguirre:2020}. One iteration of the Newton's method is performed, taking as initial seeds the roots of $f(\lambda)$, i.e.,
$$
\left(S_{-},f(S_{-})\right) = \left(u-\sqrt{gh+\a_1^2},0\right),\quad \left(S_{med},f(S_{med})\right) = \left(0,0\right),\quad \left(S_{+},f(S_{+})\right) = \left(u+\sqrt{gh+\a_1^2},0\right).
$$ 
This procedure is depicted in Figure \ref{fig:sketch_eig}. 
The approximations of the eigenvalues are the solutions of $l(\lambda) = d(\lambda)$, where $l(\lambda)$ is the straight line that is tangent to $f(\lambda)$ at the point $\left(S,f(S)\right)$, with $S=S_{-},S_{med},S_{+}$. That is, the approximated eigenvalues are
$$
\lambda^{\ast} = \dfrac{f'(S)S-f(S)+gh\left(\dqh+2\a_1\dqq\right)}{f'(S)-gh\dqq},\quad\mbox{with }S=S_{-},S_{med},S_{+}. 
$$

\begin{figure}[!ht]
	\begin{center}
		\includegraphics[width=0.7\textwidth]{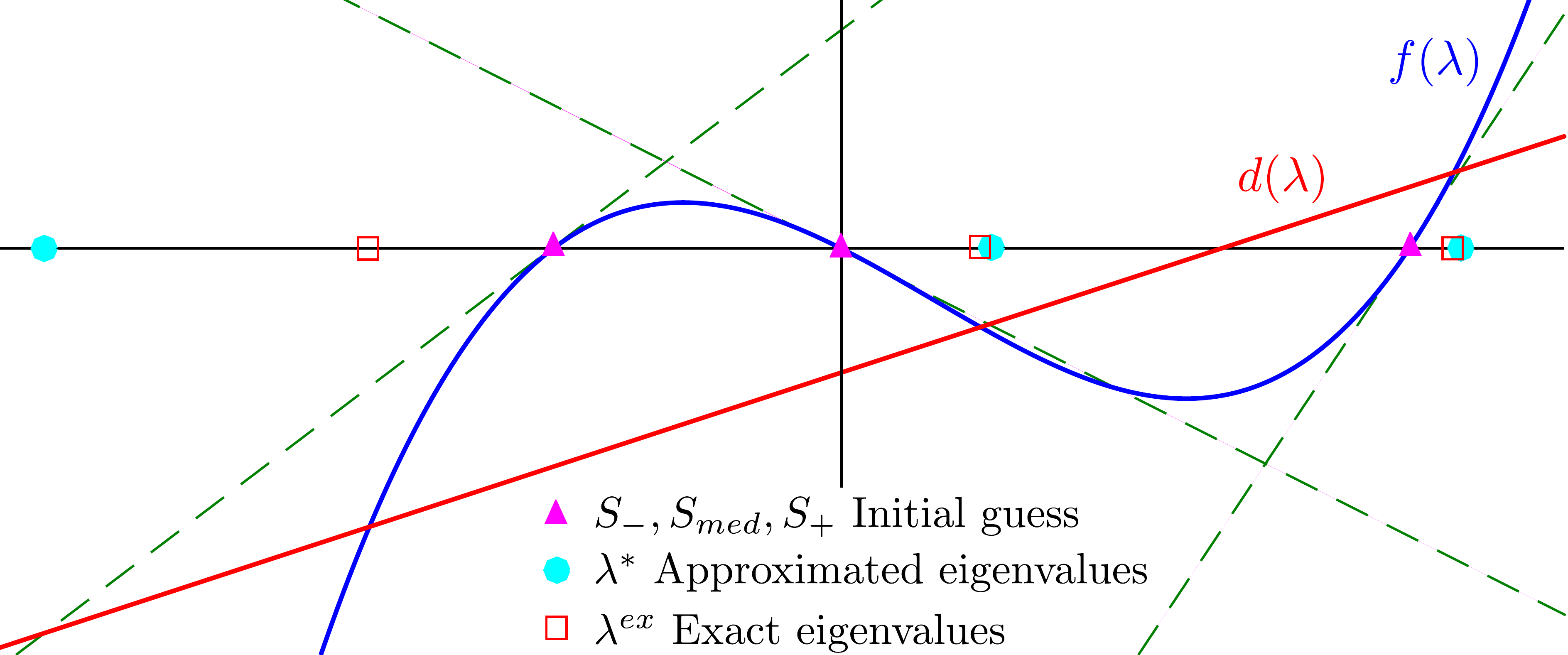}
\end{center}
	\caption{\label{fig:sketch_eig} \it{Sketch of the procedure to approximate the eigenvalues of the system matrix.}}
\end{figure}
However, one should be careful about the initial seeds, otherwise this approach might fail. This is the case for points where $u_m \approx \sqrt{gh}$, i.e. $Fr\approx 1$, or in the case of having $f'(\lambda)\approx d'(\lambda)$, i.e. $P_{S}'(\lambda) \approx 0$. Let us assume a positive velocity $u_m>0$ (the case $u_m<0$ is analogous). It is a well-known fact for the Exner system that one always has two eigenvalues of the same sign and one of opposite sign (see \cite{cordier:2011}). Looking at the shape of $f(\lambda)$ and $d(\lambda)$ for $u_m>0$, this procedure gives a positive eigenvalue for the seed $S_{+}$.  However, a negative eigenvalue is not guaranteed. For example, a configuration where this procedure fails is shown in Figure \ref{fig:sketch_eig_mod}, where three positive eigenvalues are wrongly predicted. When looking it in detail, it is the result of a bad choice of the initial seed $S_{-}$ for the lowest eigenvalue. If this seed is not properly chosen, this approach produces three positive eigenvalues, leading to spurious oscillations in the simulations. This fact will be shown in the numerical tests (see Subsection \ref{test:academic}). Assuming $S_{-} < S_{med}$, the procedure fails if $f'(S_{-})<d'(\lambda) = ghb$ holds, i.e., if the slope of $f(\lambda)$ at $S_{-}$ is lower than the slope of $d(\lambda)$. It is avoided by simply moving the initial seed $S_{-}$ to the left, i.e. modifying $S_{-}=S_{-}-\epsilon_0$, with $\epsilon_0=0.5$ for instance, as many times as necessary to have $f'(S_{-})>2 d'(\lambda)$.

\begin{figure}[!ht]
	\begin{center}
		\includegraphics[width=0.7\textwidth]{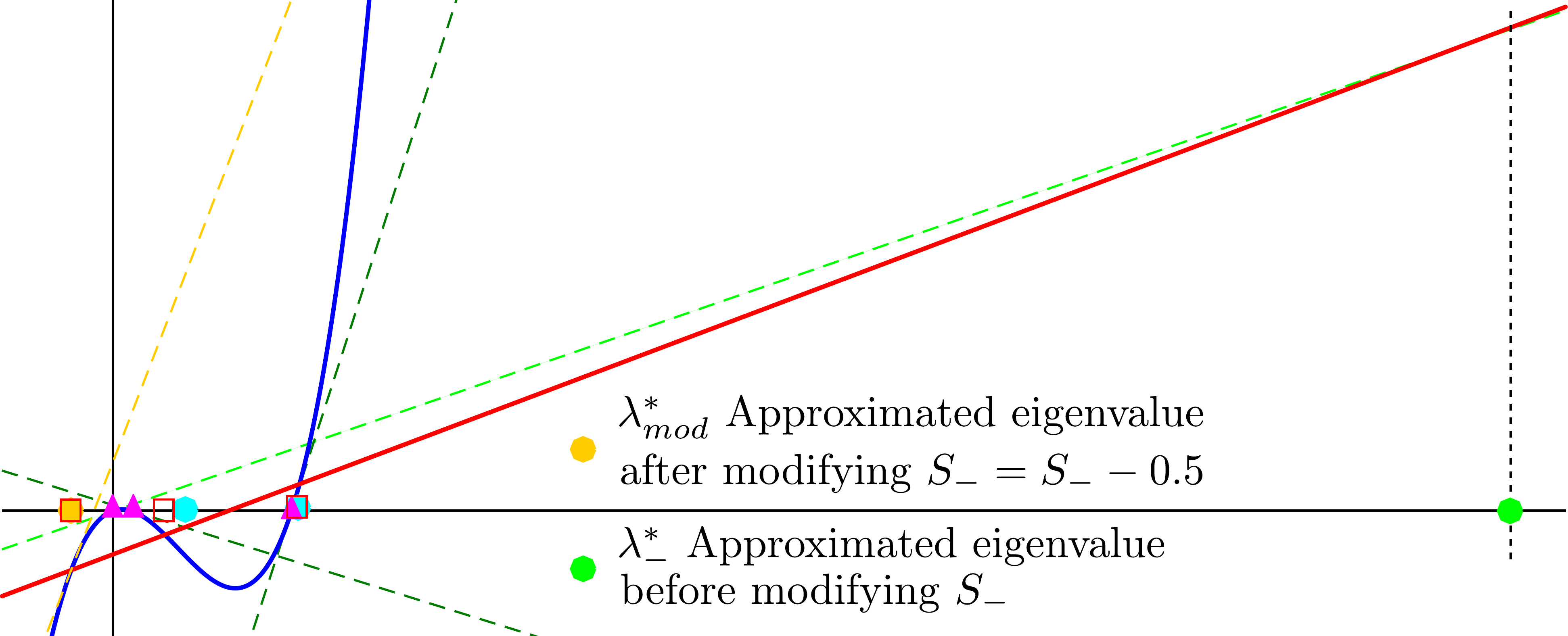}
	\end{center}
	\caption{\label{fig:sketch_eig_mod} \it{Sketch of a scenario where the procedure to approximate the eigenvalues fails.}}
\end{figure}

According to our experience, a single iteration of Newton's method is enough to have an approximation of the eigenvalues close enough to the exact eigenvalues. However, one could eventually proceed with several iterations of the Newton's method in order to have better approximations and, in particular, whenever one obtains an approximation of three eigenvalues of the same sign, which is not correct.

\bigskip 

Once these three eigenvalues are approximated, and denoting them as $\lambda_{m}<\lambda_{int}<\lambda_{p}$, the IFCP scheme should consider the internal eigenvalue with larger absolute value. This is a requirement to ensure IFCP scheme to be $L^\infty$-stable under the usual CFL condition. In order to satisfy this condition, we consider the internal eigenvalues given by Remark \ref{re:eigenMoments}, and among those values we define $\lambda_{aux}$ as the one with larger absolute value.
Now, the intermediate eigenvalue $\lambda_{med}$ in the scheme is computed as
$$
\lambda_{med} = sgn(\lambda_{m}+\lambda_{p})\max\left(|\lambda_{int}|,|\lambda_{aux}|\right).
$$

\bigskip

For the numerical tests in the following section, we consider the third-order HSWEM model, i.e. $N=3$, as one example that allows for a relatively complex vertical velocity profiles. Therefore, let us detail the eigenvalues in this particular case.
Thanks to Theorem \ref{TheoremEV}, we have the following corollary about the eigenvalues of the third-order HSWEM model:
\begin{corollary}\label{corolary}
	The eigenvalues of the third-order HSWEM model are $\lambda_{1,2,3}$, the three roots of
	$$P_{S}(\lambda)=-\lambda\left( \left(\lambda - u_m\right)^2 -gh -\alpha_1^2 \right) + gh(\dqh + \lambda \dqq + 2\alpha_1 \dqq),$$
	and the rest of the eigenvalues are explicitly given by
	$$
	\lambda_{4} = u_m,\qquad \lambda_{5,6} = u_m\pm\sqrt{\dfrac37}\a_1. 
	$$
\end{corollary}
\begin{proof}
	It follows from Theorem \ref{TheoremEV} in the particular case $N=3$, where the characteristic polynomial $\chi_{A_2} (\lambda - u_m)$ is 
	$$
	\chi_{A_2} (\lambda - u_m) = \left|
	\begin{array}{ccc}
	-\lambda + u_m & \dfrac35 \a_1  & 0 \\
	\dfrac13 \a_1 & -\lambda + u_m & \dfrac47 \a_1  \\
	0 & \dfrac25 \a_1 & -\lambda + u_m  \\
	\end{array}
	\right| = -\left(\lambda - u_m\right)\left(\left(\lambda - u_m\right)^2 - \dfrac37 \a_1^2\right).
	$$
\end{proof}
Thus, defining $\lambda_{aux}^{\pm} = u_m \pm \sqrt{\dfrac37}\a_1$, the eigenvalue $\lambda_{aux}$ is defined as
$$
\lambda_{aux} = \left\{\begin{array}{ll}
\lambda_{aux}^+ & \mbox{if } |\lambda_{aux}^+| > |\lambda_{aux}^-|,\\[3mm]
\lambda_{aux}^- & \mbox{otherwise}.
\end{array}\right.
$$
Remark that when the first-order HSWEM model, i.e. $N=1$, is considered in the numerical test, we take $\lambda_{aux} = u_m$.

\section{Numerical test}\label{se:tests}
In this section several numerical tests are performed to validate the proposed model. Our goal is to show that the moment approach allows us, in a relatively simple way, to recover the vertical structure of the flow, making it possible to improve the approximation of the velocity close to the bottom, and consequently improve the bedload transport. Concretely, two typical tests are considered depending on the time-scale of the sediment transport: the movement of a dune (weak interaction) and dam-break problems (strong water-sediment interaction). We consider academic tests showing the differences between the standard SWE model and the third-order HSWEM model, and a comparison with laboratory experiments described in \cite{spinewine:2007} (also in \cite{gonzalezAguirre:2020,juez:2013}). For all the simulations, we set the kinematic viscosity $\nu = 0.01$ $m^2/s$, and the stability condition $CFL=0.9$ is fixed. In these tests, the only mechanism generating a vertical structure is the friction term. Then, the greater the friction is, the larger the vertical structure. Nevertheless, if this friction is too large, then the movement is also too slow and therefore the vertical structure is less relevant. 

\subsection{Large-time scale: dune test}
Let us start with a simple test case, where the characteristic time of the sediment problem is much smaller than the characteristic time of the hydrodynamic problem. To this end, we consider here a dune that is swept along by a water current. We take the same configuration as in \cite{fernandezNieto:2016}, with the difference that we include the friction force between the water and the sediment layer. The initial sediment layer, height and discharge are 
$$
b(0,x) = \left\{\begin{array}{ll}
0.2 \ m& \mbox{if }x\in [4 \ m,6 \ m],\\
0.1 \ m& \mbox{otherwise},
\end{array}\right.\qquad h(0,x) = 1 \ m - b(0,x)\ ,\qquad q(0,x) = 1.5\  m^2/s,
$$
and the sediment properties are
$$
\rho/\r_s = 0.34,\quad d_s = 1\ mm, \quad\theta_c = 0.047,\quad n=0.01,\quad \varphi = 0.95.$$ 
Now, the moment approach allows us to use different vertical profiles of the velocity at initial time (parabolic, linear or constant), all of them verifying 
$$
q(0,x) = \dint_0^1 h(0,x)u(0,x,\xi)d\xi = 1.5\ m^2/s,  
$$
that is, the mean discharge is equal for all the initial profiles. We take three different velocity profiles (see Figure \ref{fig:dune_t150})
\begin{equation}\label{eq:ini_vel_dune}
u(0,x,\xi) = a_0 + a_1\xi + a_2\xi^2\ m/s,
\end{equation}
given by
\begin{subequations}
\begin{equation}
a_2=0\ m/s,\quad a_1=0\ m/s,\quad a_0 = 1.5/h(x)\ m/s,\tag{constant case}
\end{equation}
\begin{equation}
a_2= 0\ m/s,\quad a_1 = 0.5\ m/s, \quad a_0=1.5/h(x)-a_1/2\ m/s,\tag{linear case}
\end{equation}
\begin{equation}
a_2 = 0.35\ m/s,\quad a_1=11/15\ m/s,\quad a_0 = 1.5/h(x)-a_2/3-a_1/2\ m/s.\tag{parabolic case}
\end{equation}
\end{subequations}

For this test we take $x\in[0,10]\ m$ with $800$ points, and subcritical boundary conditions are considered. The profile \eqref{eq:ini_vel_dune} is imposed upstream and the water height $h(t,10) = 1 - b(0,10)\ m$ is assumed downstream.

Let us remark that the case of very slow processes are characterized by a long final time in the simulation. Then, these problems require an accurate time discretization. Otherwise the numerical diffusion could hide the results. Thus, we consider in this case the third-order scheme defined by \eqref{eq:scheme}-\eqref{eq:diffusion} combined with a third order CWENO reconstructed states technique \cite{cravero2019optimal}, and the SI-RK3 time integrator \cite{chertock:15}. 

\medskip

\begin{figure}[!ht]
	\begin{center}
		\includegraphics[width=1.\textwidth]{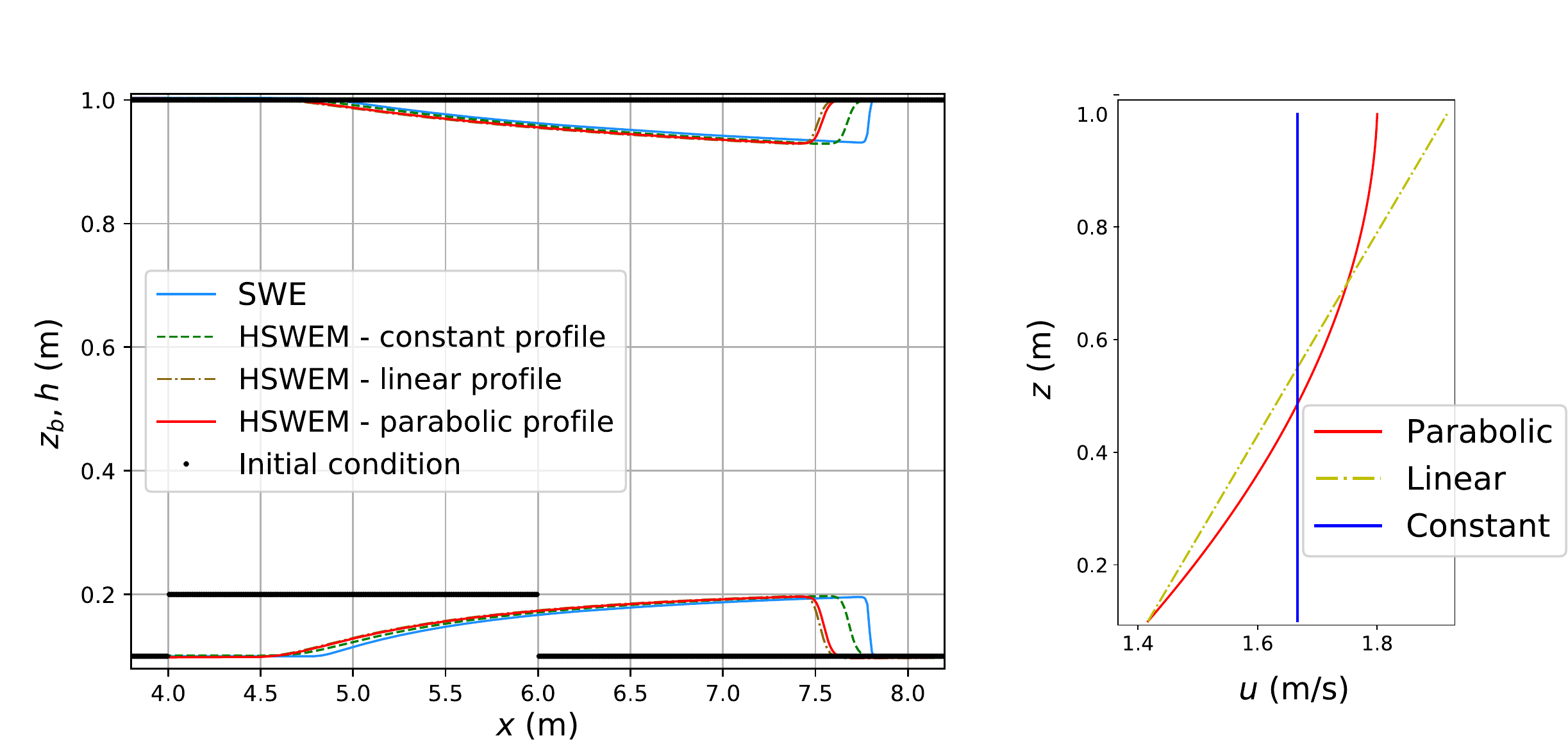}
	\end{center}
	\caption{\label{fig:dune_t150} \it{Left: Water surface and sediment layer at time $t=150$ s, for the SWE (solid blue line) model and the third-order HSWEM model with constant (dashed green line), linear (dot-dashed brown) and parabolic (solid red line) initial profile of velocity. Black symbols represent the initial condition. Right: Initial profiles of velocity.}}
\end{figure}

\begin{figure}[!ht]
	\begin{center}
		\includegraphics[width=0.85\textwidth]{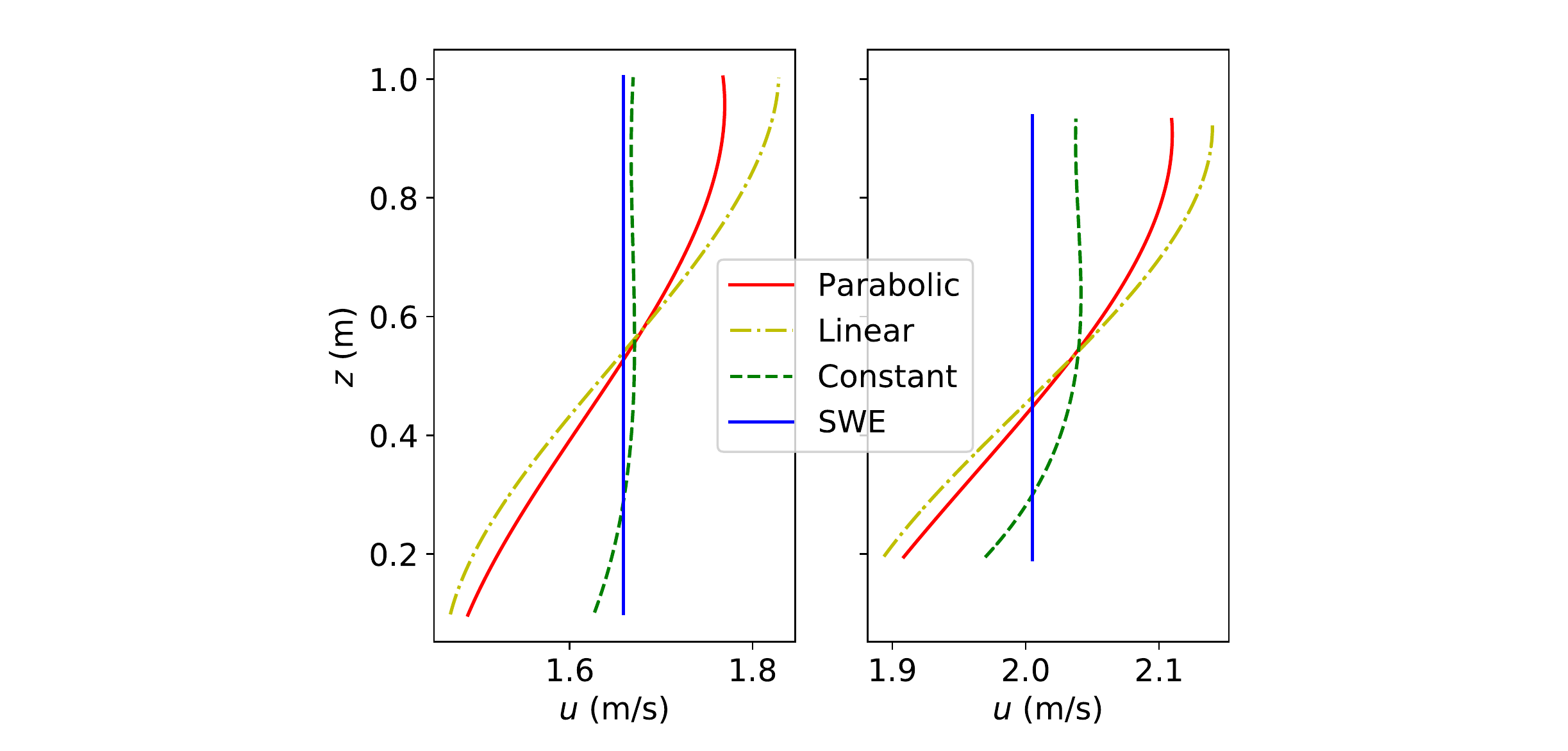}
	\end{center}
	\caption{\label{fig:dune_u_profiles} \it{Vertical profiles of velocity obtained at $t=150$ s, at point $x=3$ m (left hand side), and $x=7.25$ m (right hand side), with the SWE (solid blue line) and HSWEM models with constant (dashed green line), linear (dot-dashed brown) and parabolic (solid red line) initial profile of velocity.}}
\end{figure}
\begin{figure}[!ht]
	\begin{center}
		\includegraphics[width=1.\textwidth]{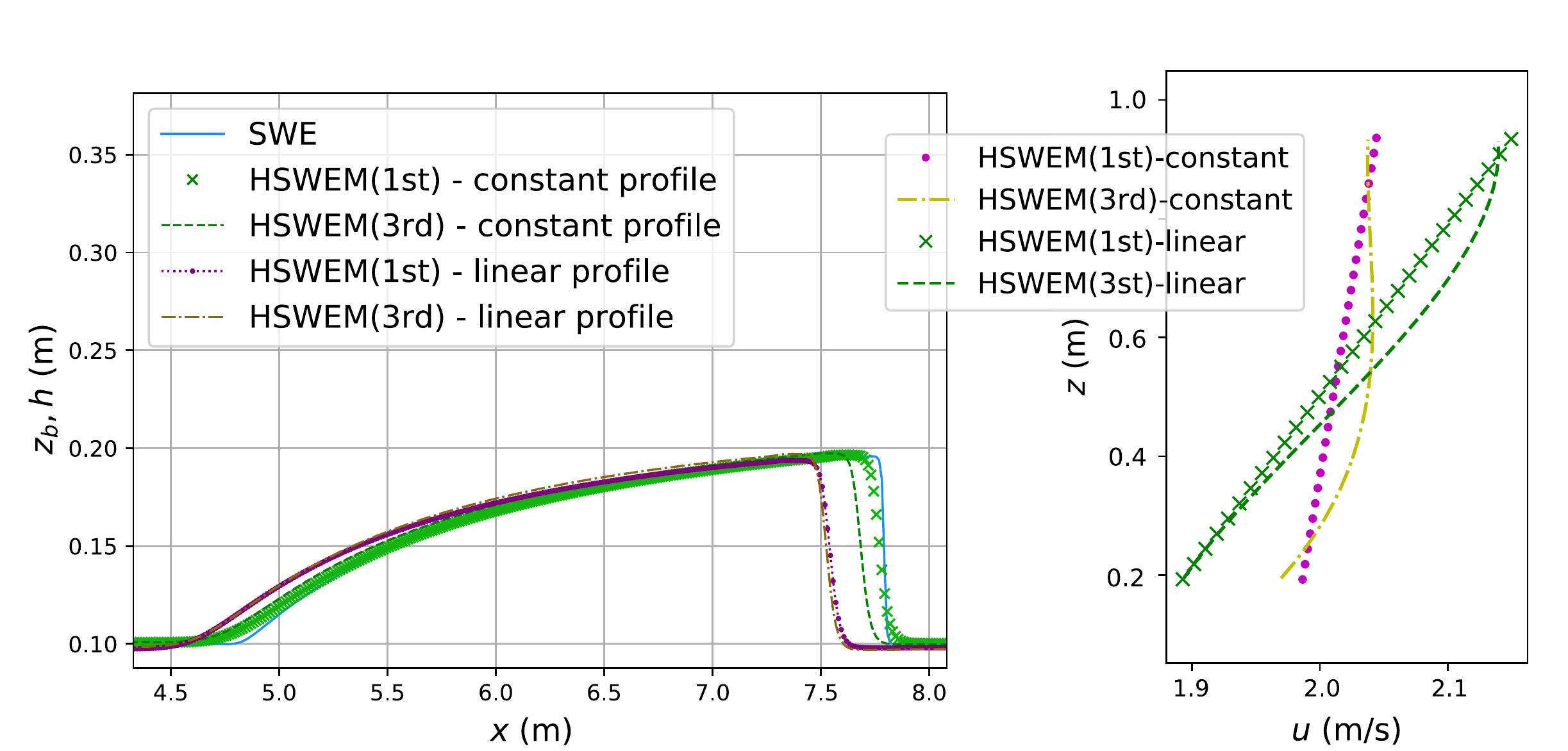}
	\end{center}
	\caption{\label{fig:dune_t150_hswm1} \it{Left: Sediment layer (left hand side) and vertical profile of velocity at $x=7.25$ m (right hand side), at time $t=150$ s, for the SWE (solid blue line) model, the third-order HSWEM model with constant (dashed green line), linear (dot-dashed brown) initial profile of velocity, and the first-order HSWEM model with constant (green crosses), linear (magenta circles) initial profile of velocity.}}
\end{figure}

Figure \ref{fig:dune_t150} shows the height and the sediment layer at final time, for the SWE model and third-order HSWEM model, for the different initial profiles of the velocity. Remark that a constant profile is the only possibility for the SWE model, whereas the HSWEM model allows us to define a constant, linear or parabolic profile of the velocity at initial time. We see that the position of the front for the HSWEM model is delayed when compared to the SWE model and for all initial velocity profiles. This is due to the fact that the friction with the sediment layer implies that the near bottom velocity decreases, as shown in Figure \ref{fig:dune_u_profiles}. Remark that, even though a constant  initial profile is prescribed, since we are using the third-order HSWEM model, after some time the velocity profile is no longer constant and a vertical structure appears. Despite the fact that the velocity computed with the HSWEM model is greater than the one computed with the SWE model in the upper half of the water column, it is the opposite near the bottom. Therefore, the movement of the sediment is slower for the HSWEM model. 

In addition, we see that the velocity near the bottom is slower if a linear or parabolic initial profile is used. In these cases, the sediment moves even slower since the computed velocity at the bottom is also smaller. The vertical structure reproduced by the third-order HSWEM model is similar in both cases. Therefore, the solutions only slightly differ, which is a consistent result. It could be interpreted as a convergence on the velocity profile. \\

In Figure \ref{fig:dune_t150_hswm1} we compare the third-order and first-order HSWEM models, with constant and linear initial profile of velocity. We see that for the linear initial profile, both models produce similar results, although the vertical profile of velocity is not linear but parabolic for the third-order HSWEM model after some time. More differences are found for a constant initial profile. We see that the results of the first-order model are similar to the SWE model, whereas the third-order model differs from these ones. Therefore, for constant initial profiles, the third-order model notably improve the results of the SWE and first-order HSWEM models.

\subsection{Short time scale: dam-break test}
We investigate now configurations where important changes occur at short times. Concretely, dam-break configurations are considered. First we show an academic test and secondly some comparisons with laboratory experiments are performed. For all the configurations, as usual in dam-break problems, it starts from water at rest, i.e. $q(0,x) = 0\ m^2/s$, where all $u_m,\a_1,\a_2,\a_3$ are zero. Note that imposing an initial vertical profile of velocity different from this one would lead to no physically relevant scenarios. In the following tests, free boundary conditions are considered. 

Let us firstly show an academic test, and secondly a variety of comparisons with laboratory experiments. 

\subsubsection{Academic dam-break test}\label{test:academic}
We consider here a dam-break problem in a simple configuration. The initial height is
$$
h(0,x) = \left\{\begin{array}{ll}
1\ m & \mbox{if }x<0,\\
0.05\ m & \mbox{otherwise}.
\end{array}\right. ,
$$
and the sediment is fixed as $b(0,x) = 0 \ m$. In this case, the sediment properties are
$$
\rho= 1000\ Kg/m^2,\quad \r_s = 1580\  Kg/m^2,\quad d_s = 3.9\  mm, \quad\theta_c = 0.047,\quad n=0.0365,\quad \varphi = 0.47,$$
and we take $x\in [-6,6]\ m$, with $1200$ points.

\begin{figure}[!ht]
	\begin{center}
		\includegraphics[width=0.65\textwidth]{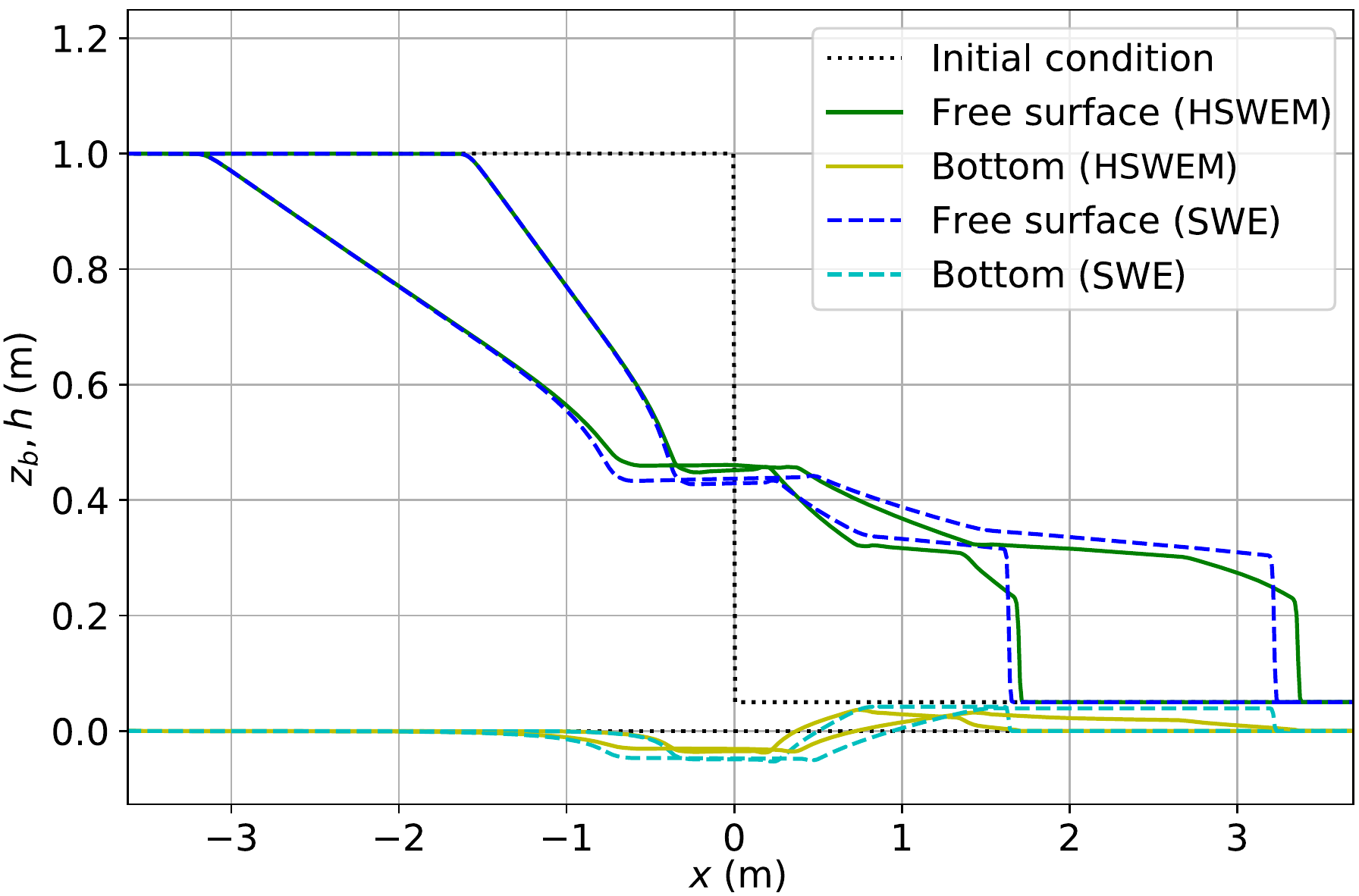}
	\end{center}
	\caption{\label{fig:dam_break_academic} \it{Water surface and sediment layer at times $t=0.5 \ s, 1 \ s$ computed with the HSWEM (solid lines) and SWE (dashed lines) models. Dotted lines represent the initial condition.} }
\end{figure}

\begin{figure}[!ht]
	\begin{center}
		\includegraphics[width=0.66\textwidth]{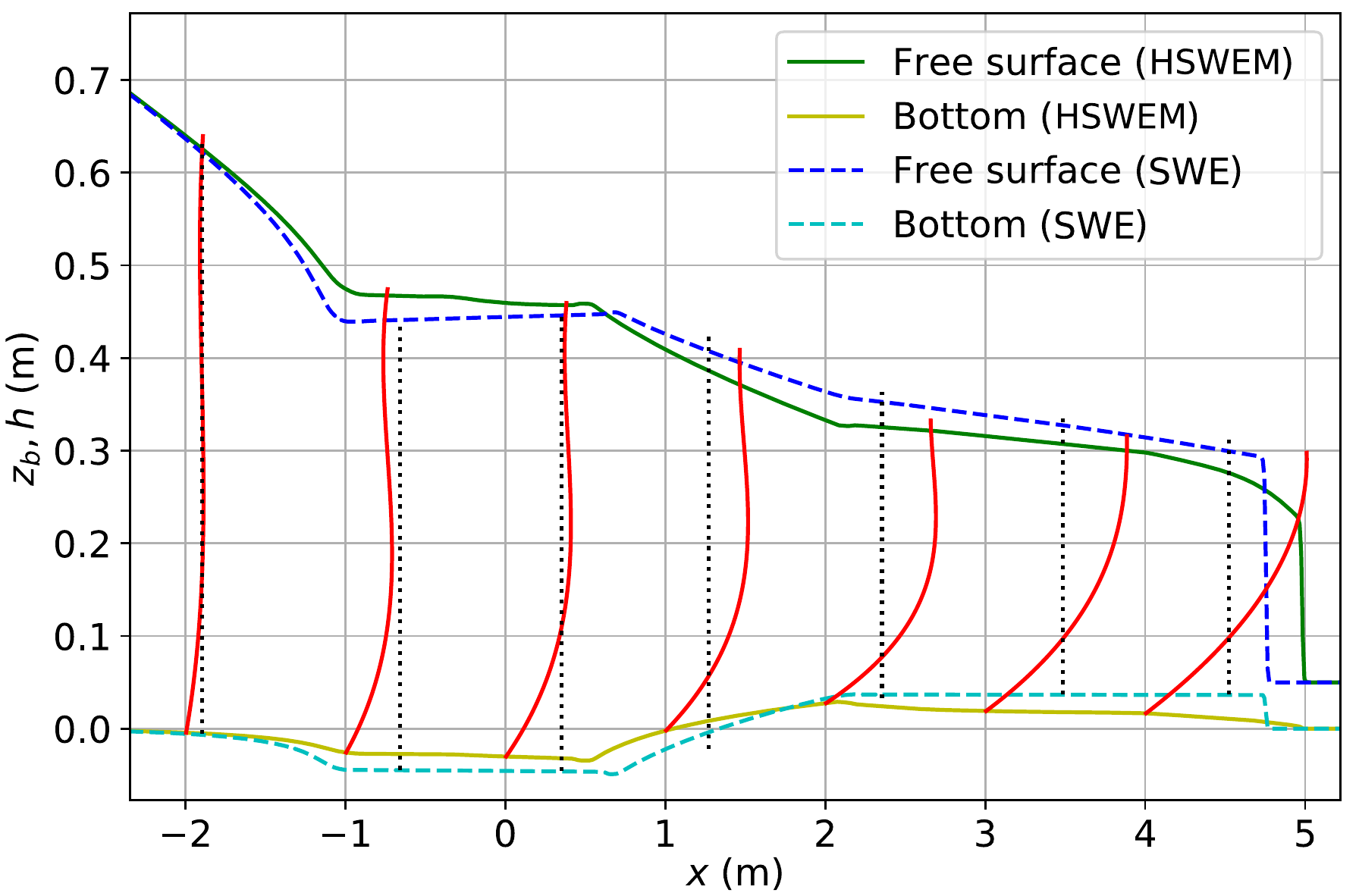}
	\end{center}
	\caption{\label{fig:dam_break_academic_u} \it{Water surface and sediment layer at time $t=1.5$ s computed with the HSWEM  (solid lines) and SWE (dashed lines) models. Solid red (resp. dotted black) lines are the vertical profiles of velocity computed with the HSWEM (resp. SWE) model at points $x=-2,-1,...,4$ m.}}
\end{figure}

Figure \ref{fig:dam_break_academic} shows the water height and the sediment layer at different times. We see that the approximations of both the free surface and the sediment layer with the third-order HSWEM model differ from the approximations given by the SWE system. In particular, they are different close to the front, where both the velocity and the vertical structure (i.e. the deviation of a constant profile) are greater. Focusing on the sediment layer, we see that the HSWEM model produces a more realistic shape of the bottom, removing the vertical column appearing with the SWE model. This is also shown in Figure \ref{fig:dam_break_academic_u}, where the vertical profiles of velocity at several points are depicted, at time $t=1.5$ s. The HSWEM model allows us to recover the typical parabolic profile of the velocity, namely near the front. Then, the velocity at the bottom, that is the one used to transport the sediment, is smaller than the averaged velocity computed with the SWE model. This effect decreases the bedload transport, notably improving the shape of the sediment layer. However, the velocity along almost the whole water column is greater with the HSWEM model, leading to a more advanced position of the front. Concretely, the mean velocity computed with the HSWEM model is approximately 7$\%$ greater than the velocity obtained with the SWE model, for $x\geq1$ m at $t=1.5$ s. This is a qualitatively correct and expected behaviour.\\

Here, we also show how the approximation of the eigenvalues described in Section \ref{se:num_aprox} fails if no modification of the procedure is taken into account, together with the same configuration taking care of the initial seed for the Newton method approximating the lowest eigenvalue. In Figure \ref{fig:dam_break_academic_eigen_mod} we show this comparison, where the standard procedure gives three positive eigenvalues in some nodes, leading to the appearance of spurious oscillations in the simulation. When the proposed correction is considered it is possible to ensure at least one positive and one negative eigenvalue and the larger and smaller eigenvalues are effectively bounds of the exact eigenvalues. The stability condition is then fulfilled and all oscillations disappear, showing the efficiency of the proposed numerical method. 

\begin{figure}[!ht]
	\begin{center}
		\includegraphics[width=0.49\textwidth]{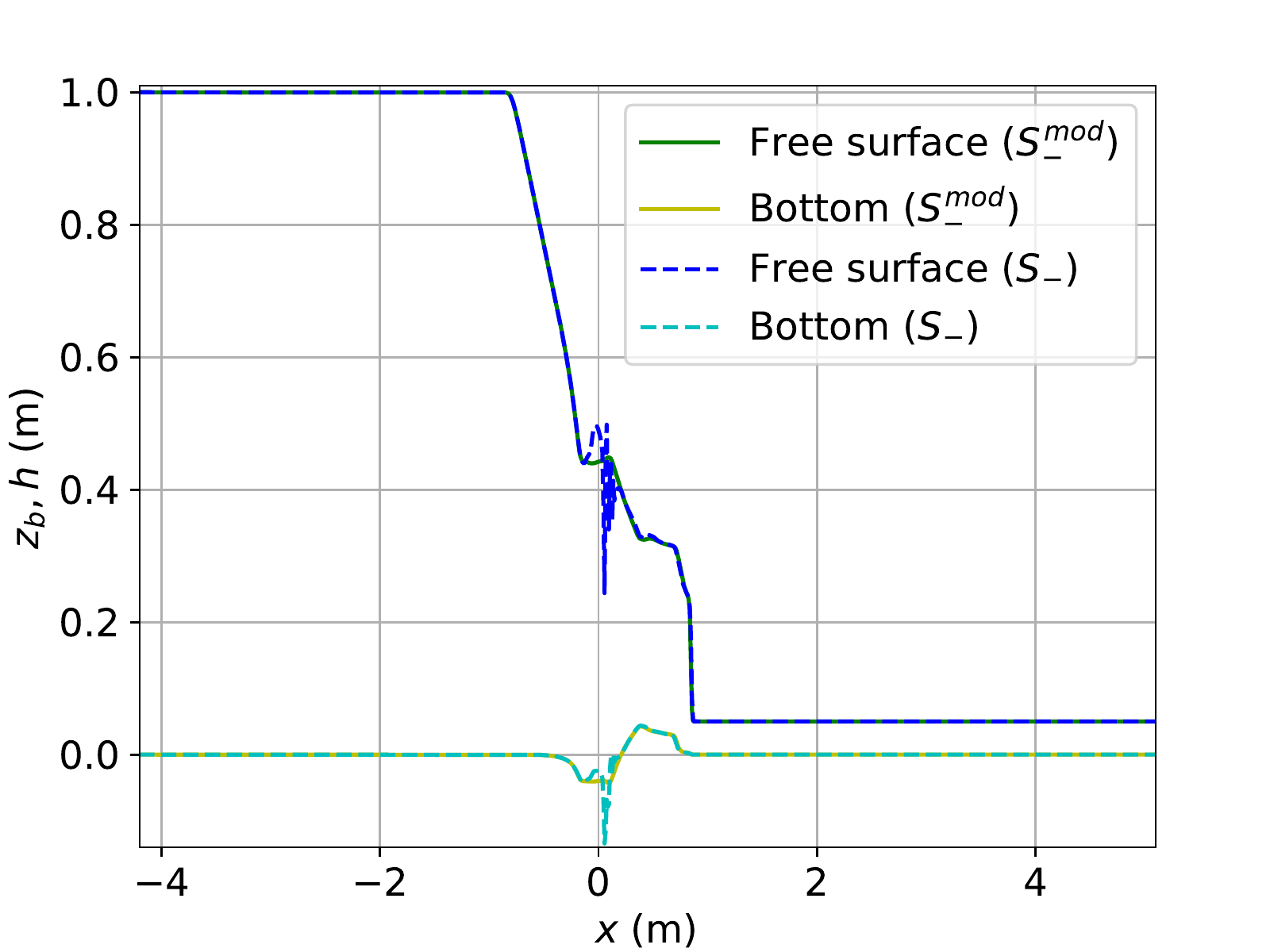}
		\includegraphics[width=0.49\textwidth]{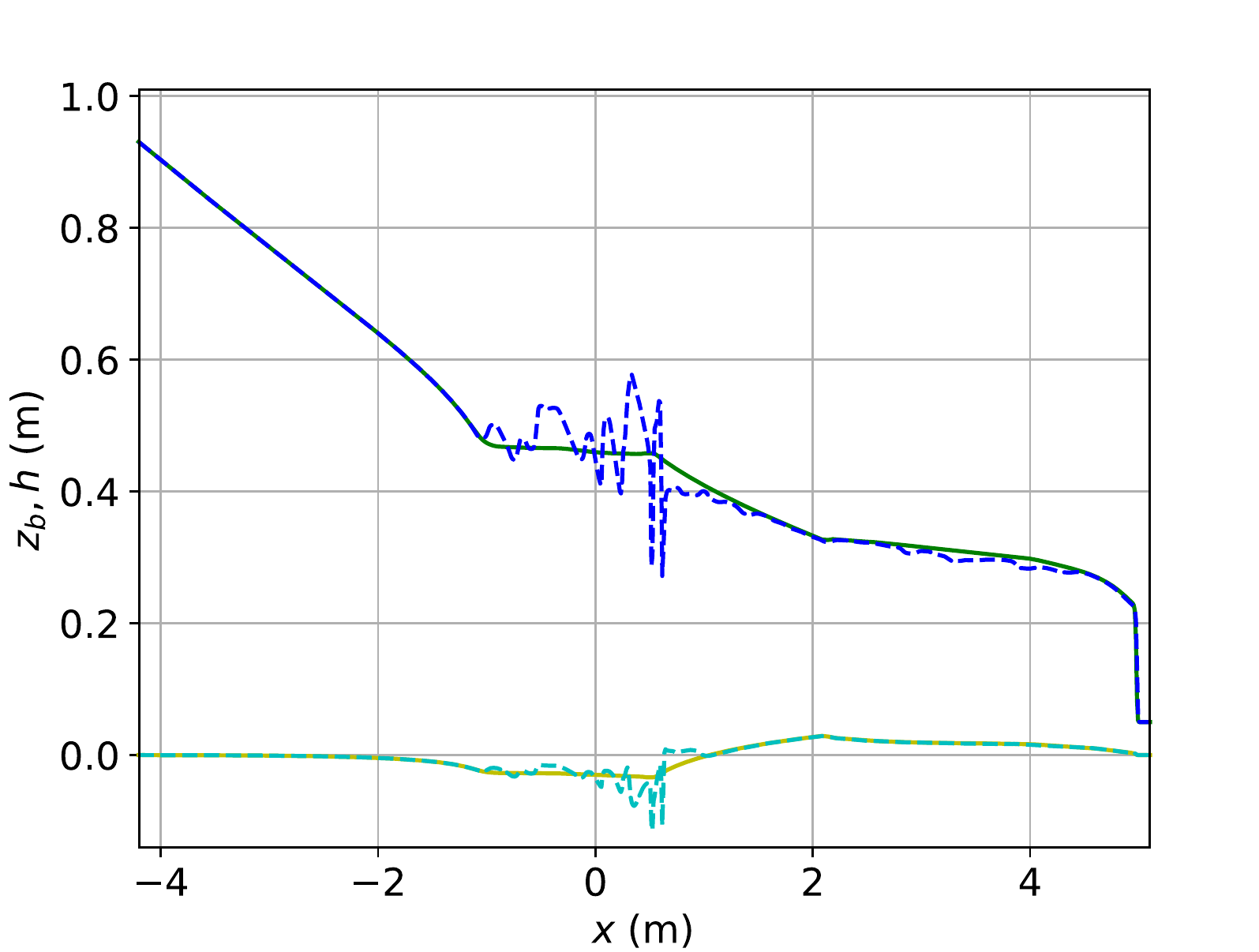}
	\end{center}
	\caption{\label{fig:dam_break_academic_eigen_mod} \it{Free surface and bottom profiles at times $t=0.25,0.5,1,1.5$ $t=0.25,1.5$ s computed with the HSWEM model modifying the initial seed for the lowest eigenvalue when needed (solid lines), and without modifying it (dashed lines).} }
\end{figure}

\subsubsection{Laboratory dam-break experiments}

Now, we compare our results with laboratory experiments of dam-break problems in \cite{spinewine:2007}, where two different materials are used, uniform coarse sand and extruded PVC pellets.\\

\underline{Experiment 1 (PVC pellets):} First we consider the experiments with PVC pellets, for which the initial height is given by
\begin{equation}\tag{Exp 1}\label{eq:exp1}
h(0,x) = \left\{\begin{array}{ll}
0.35\ m & \mbox{if }x<0,\\
0\ m & \mbox{otherwise},
\end{array}\right.
\end{equation}
and the bottom is set to $b(0,x) = 0\ m$.  The sediment properties are
$$
\rho= 1000\ Kg/m^2,\quad \r_s = 1580\ Kg/m^2,\quad d_s = 3.9\ mm, \quad\theta_c = 0.047,\quad n=0.0165,\quad \varphi = 0.47.$$
In this case, for the Manning friction law we consider $n=0.0365$. We take $x\in [-3,3] \ m$ discretized by $500$ points.

\begin{figure}[!ht]
	\begin{center}
		\includegraphics[width=0.75\textwidth]{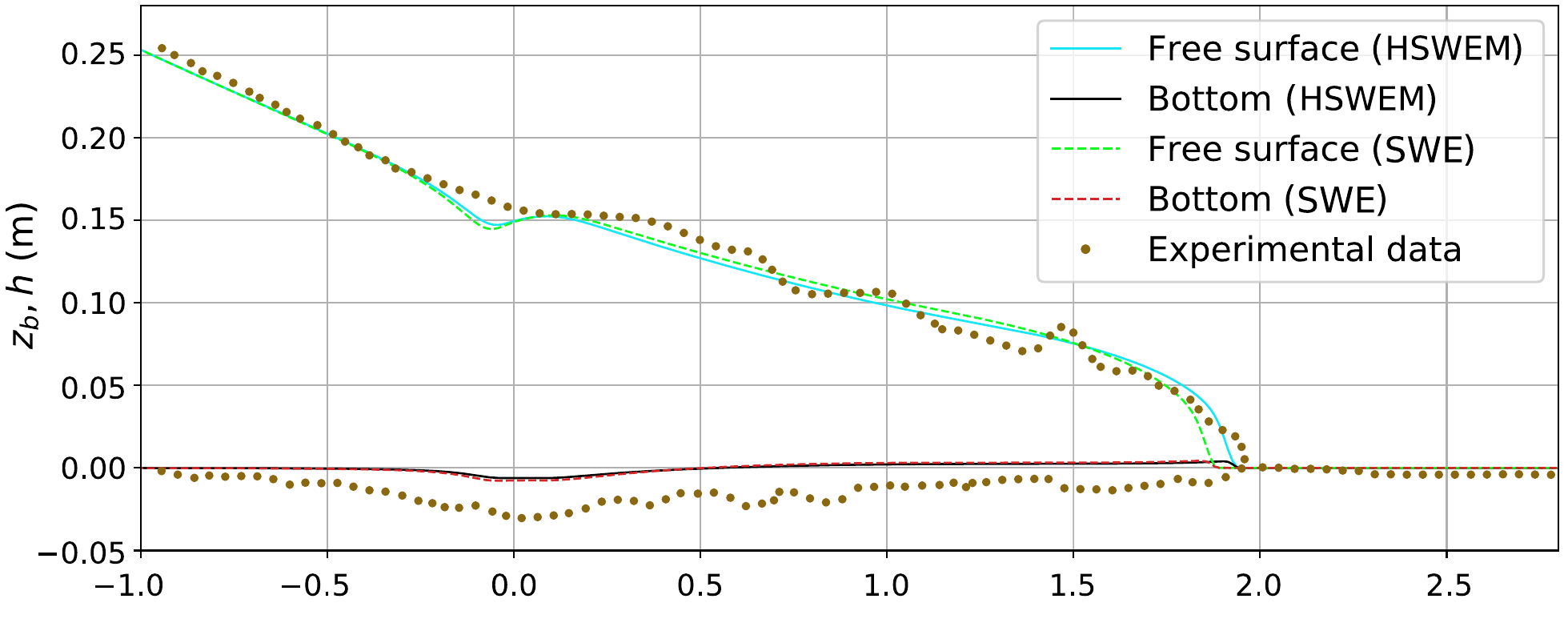}
		\includegraphics[width=0.75\textwidth]{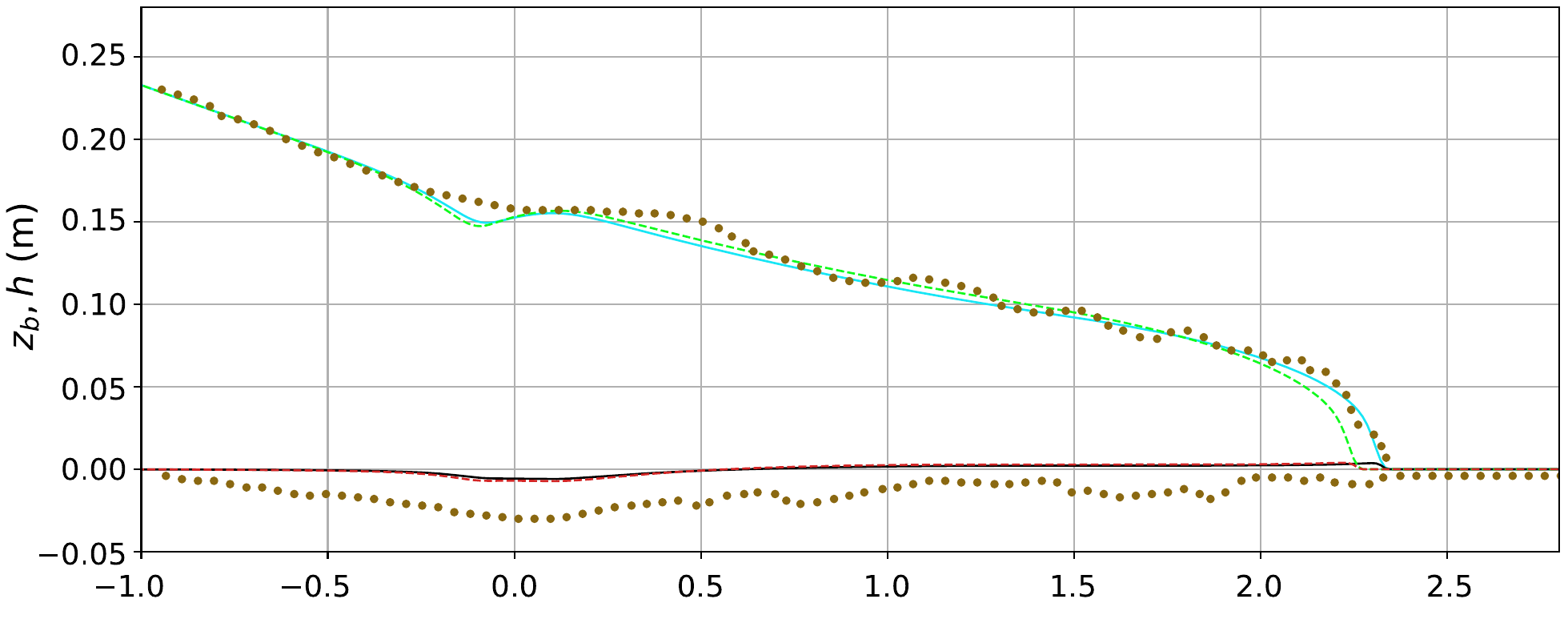}
		\includegraphics[width=0.75\textwidth]{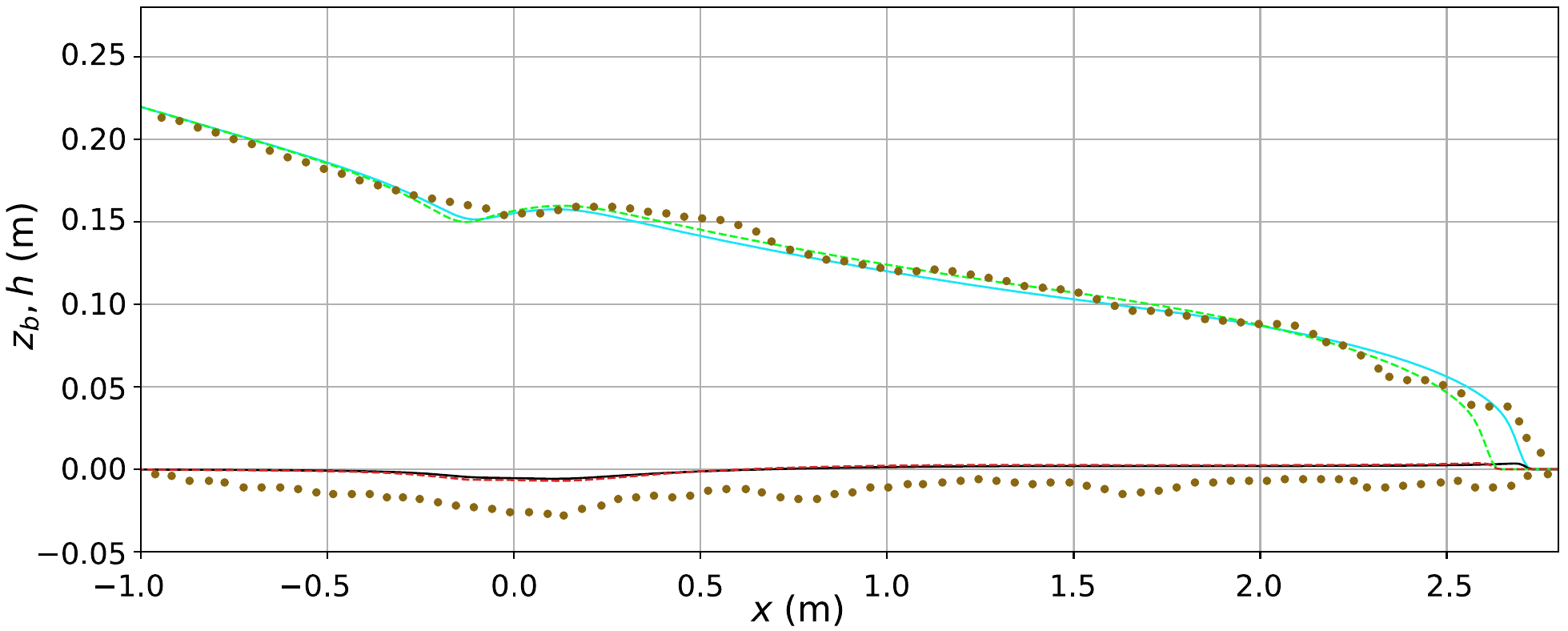}
	\end{center}
	\caption{\label{fig:exp3_data} \it{Dam break experiments and simulations at times $t=1,1.25,1.5$ s, for configuration \eqref{eq:exp1} (experiment 3 in \cite{gonzalezAguirre:2020}). } }
\end{figure}
\begin{figure}[!ht]
	\begin{center}
		\includegraphics[width=0.75\textwidth]{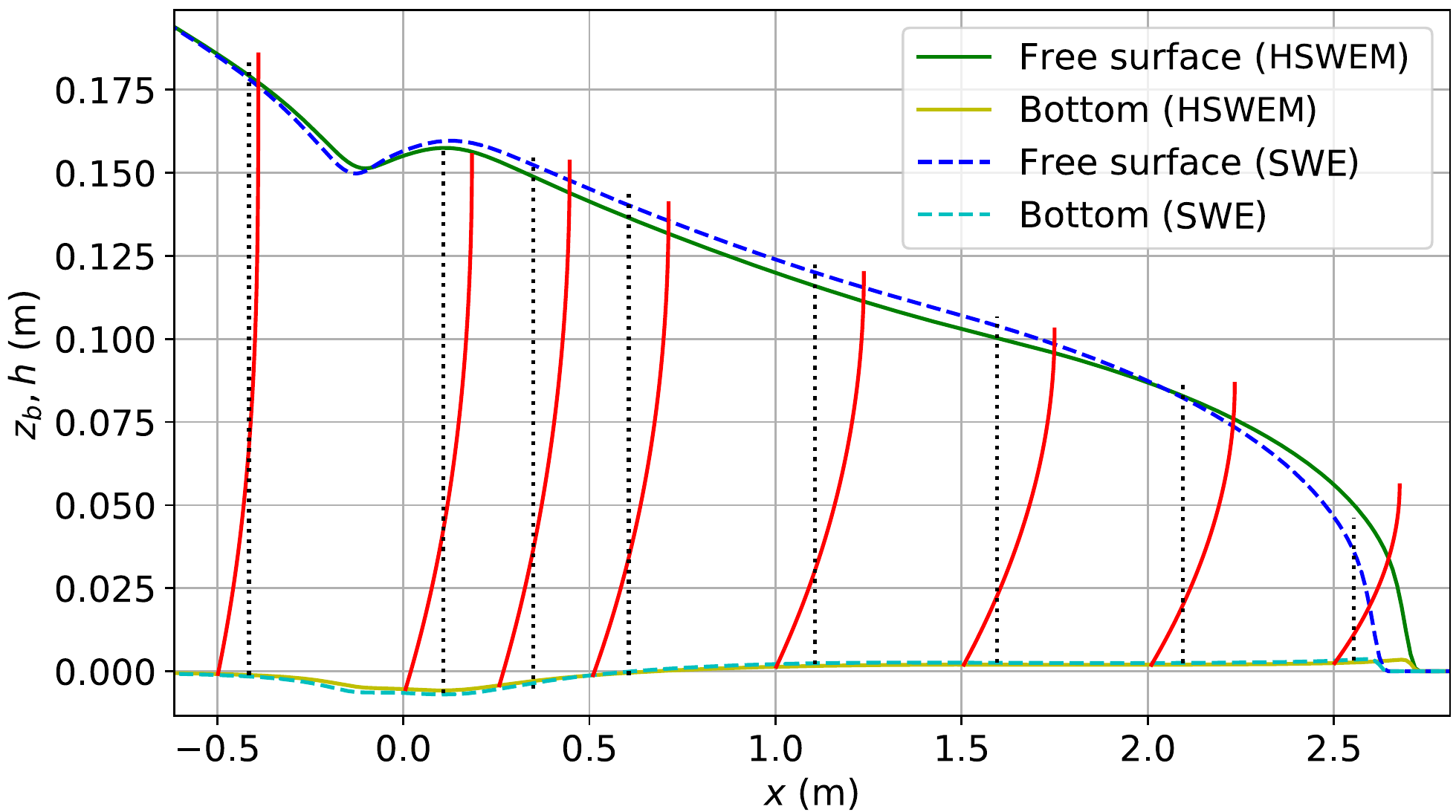}
		\end{center}
	\caption{\label{fig:exp3_u} \it{Dam break experiments and simulations at time $t=1.5$ s, for configuration \eqref{eq:exp1}, with vertical profiles of velocity at $x=-0.5\ m,0\ m,...,2.5\ m$, computed with the HSWEM (solid red lines) and SWE (dotted black lines) models.}}
\end{figure}

Figure \ref{fig:exp3_data} shows that the HSWEM model significantly improves the position of the moving front in comparison to the standard SWE model. However, the rest of the solution shows almost no difference to the SWE model. The bottom topography is not accurately captured by both models, due to neglecting erosion, deposition and suspension. As the SWE model has been chosen as starting point for the HSWEM model, we are only considering bedload sediment transport while in this test the erosion, deposition and suspended sediment transport are relevant (see for example \cite{gonzalezAguirre:2020}). Consequently, we cannot expect the model to exactly recover the erosion produced at the bottom in the laboratory results. Our focus is on a more accurate description of the vertical velocity profile, which will indeed result in an improvement description of the bedload transport (see Figure \ref{fig:exp3_u}). 
In Figure \ref{fig:exp3_u} we see the same behaviour as in previous test, that is, the computed mean velocity near the front is greater with the HSWEM model than with the SWE model, leading to a faster and more accurate movement of the water surface. On the contrary, the velocity close to the bottom is smaller with the HSWEM system.

\bigskip

\underline{Experiment 2, 3 (coarse sand):}
In this case two different configurations are considered  (\ref{eq:exp2}, \ref{eq:exp3}), where the initial height is given by

\medskip

	\begin{equation}\tag{Exp 2}\label{eq:exp2}
	h(0,x) = \left\{\begin{array}{ll}
	0.25\ m & \mbox{if }x<0,\\
	0\ m & \mbox{otherwise},
	\end{array}\right.
	\end{equation}
		\begin{equation}\tag{Exp 3}\label{eq:exp3}
	h(0,x) = \left\{\begin{array}{ll}
	0.25\ m & \mbox{if }x<0,\\
	0.1\ m & \mbox{otherwise},
	\end{array}\right.
	\end{equation}

\medskip

\noindent and the bottom is $$
b(0,x) = \left\{\begin{array}{ll}
0.1 \ m& \mbox{if }x<0,\\
0 \ m& \mbox{otherwise}.
\end{array}\right.
$$ 
In this case, the sediment properties are
$$
\rho= 1000\ Kg/m^2,\quad \r_s = 2683\ Kg/m^2,\quad d_s = 1.82\ mm, \quad\theta_c = 0.047,\quad n=0.0165,\quad \varphi = 0.47.$$
For the Manning friction law we consider $n=0.0165$ and we take $x\in [-1.5,3.5]\ m$ with $500$ points.

Figure \ref{fig:juez_zb} shows the results at time $t=1.5$ s. We do not see differences between the HSWEM and the SWE models. This is due to the fact that the friction is not large enough to generate a vertical structure of the flow. This is shown in Figure \ref{fig:juez_zb_perfiles}, where we see the vertical profiles of velocity along the $x$ direction. It shows that the vertical profiles are very close to a constant profile in almost the whole domain. Only for $x\in[0,0.5]$ there is some vertical structure, but it is not large enough to produce significant differences between the results obtained with both models. 

In Figure \ref{fig:juez_zb_perfiles_FR}, we have increased the Manning coefficient, only for the friction term, to $n=0.0365$. This is to show the impact of the friction term on the vertical structure of the flow. In that case, differences between both models appear. The vertical structure of the fluid is now more relevant. Notice that this is especially the case in areas where the Froude number ($u/\sqrt{gh}$) is large.

\begin{figure}[!ht]
	\begin{center}
		\includegraphics[width=0.75\textwidth]{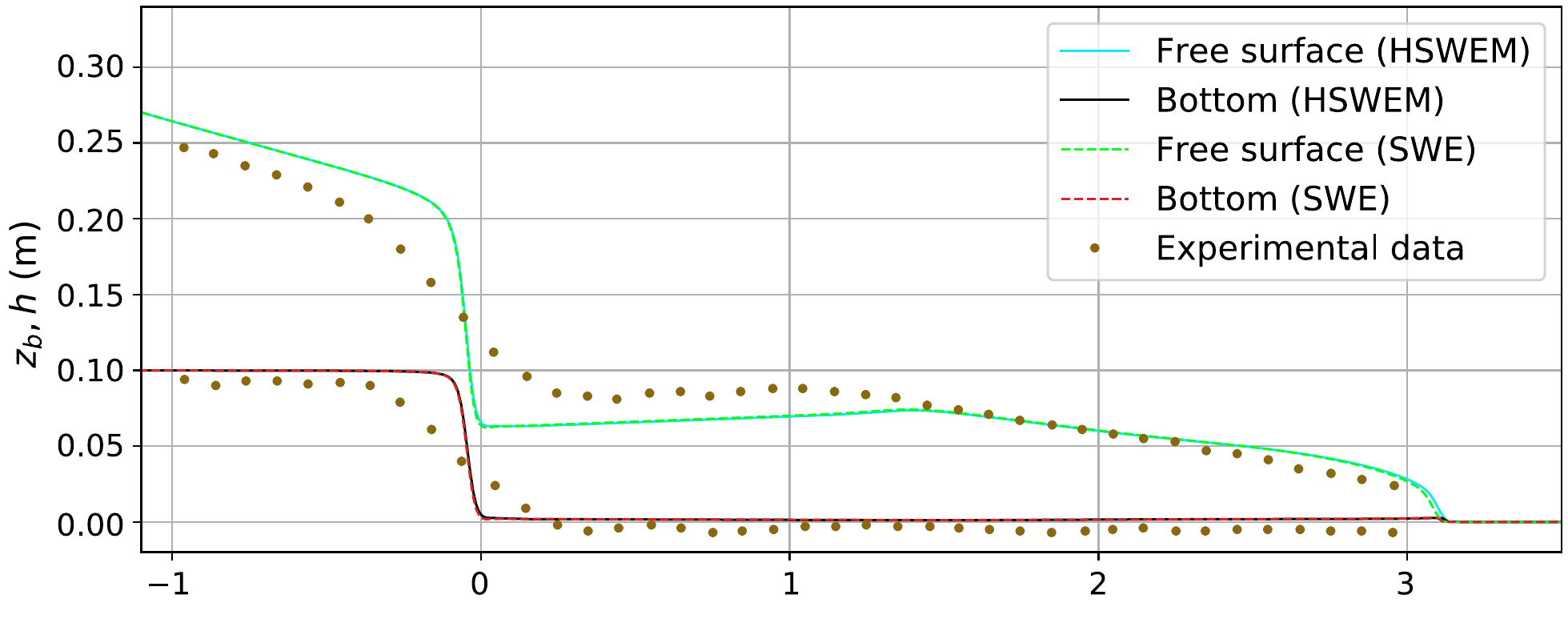}
		\includegraphics[width=0.75\textwidth]{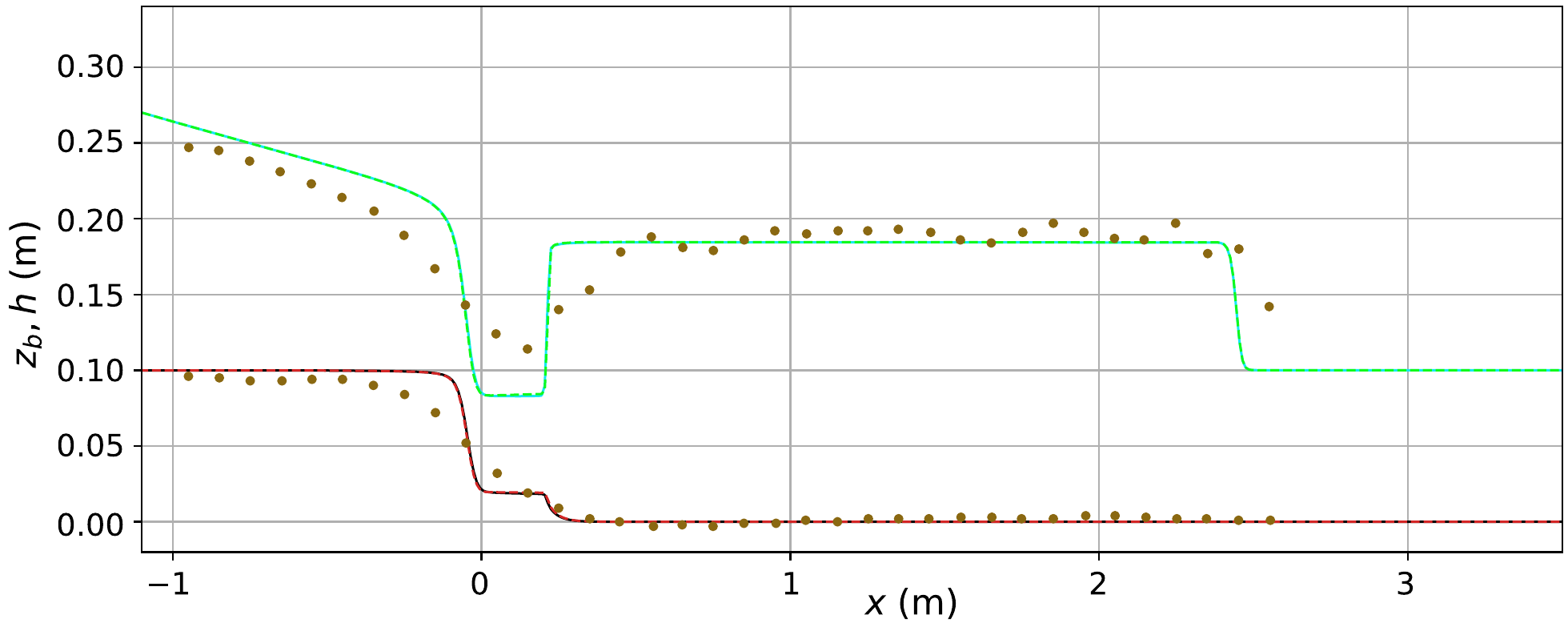}
	\end{center}
	\caption{\label{fig:juez_zb} \it{Dam break experiments and simulations at time $t=1.5$ s, for configurations \ref{eq:exp2} and \ref{eq:exp3} (test case A and B in \cite{juez:2013}). }}
	\end{figure}

	\begin{figure}[!ht]
		\begin{center}
			\includegraphics[width=0.75\textwidth]{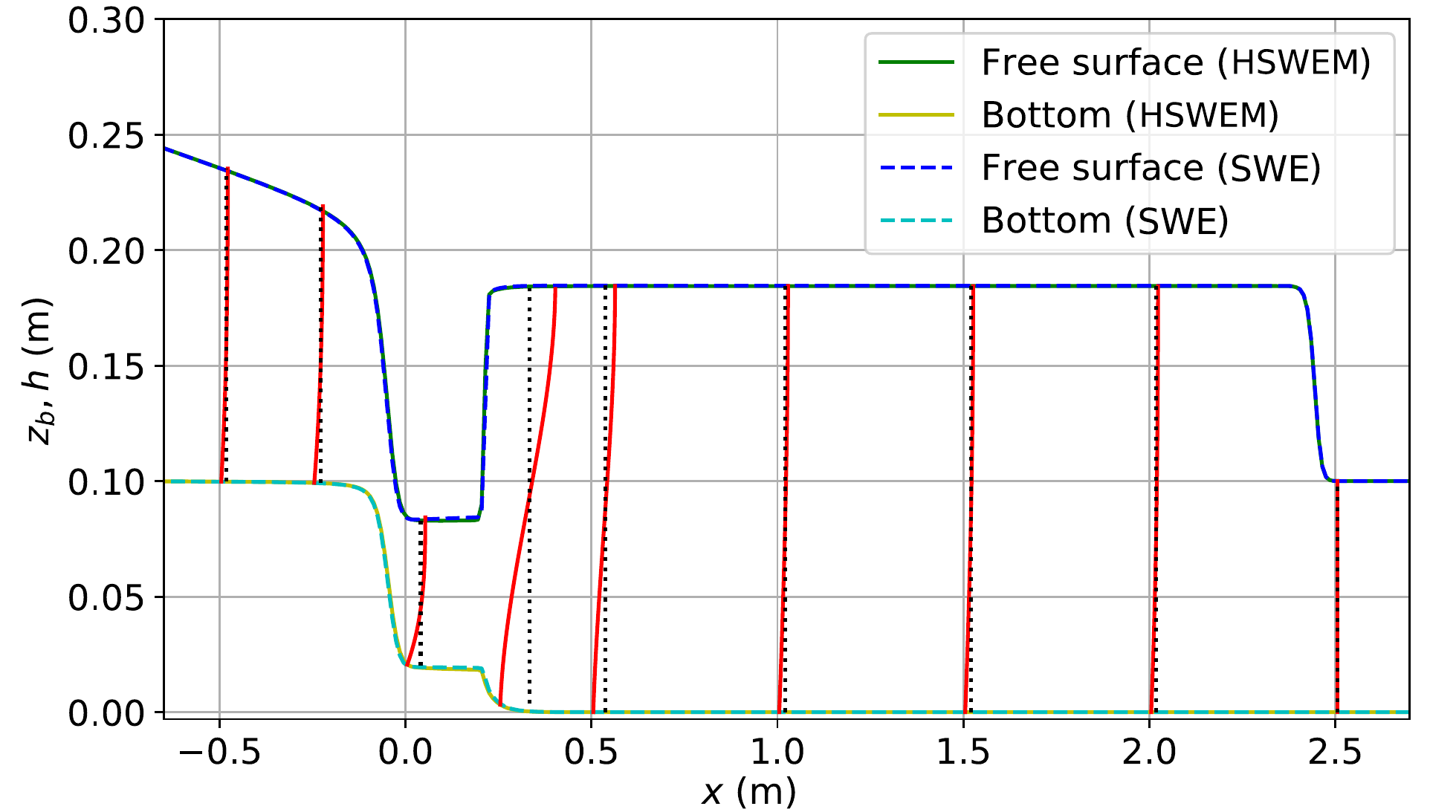}
		\end{center}
		\caption{\label{fig:juez_zb_perfiles} \it{Dam break experiments and simulations at time $t=1.5$ s, with vertical profiles of velocity at $x=-0.5,-0.25,0,0.25,0.5,1,...,2.5$ m, computed with the HSWEM (solid red lines) and SWE (dotted black lines) models.} }
	\end{figure}

\begin{figure}[!ht]
	\begin{center}
		\includegraphics[width=0.75\textwidth]{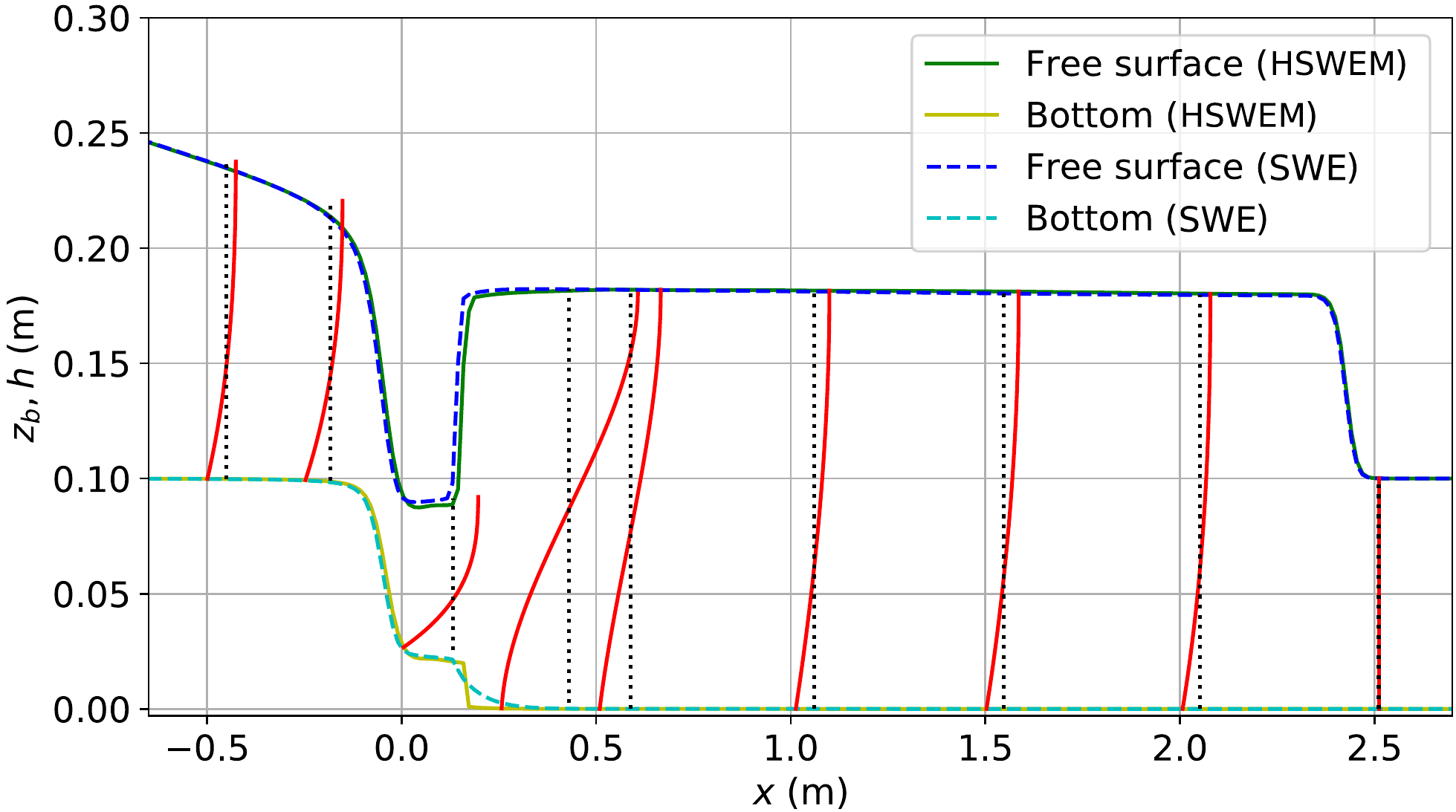}\\[3mm]
		\includegraphics[width=0.75\textwidth]{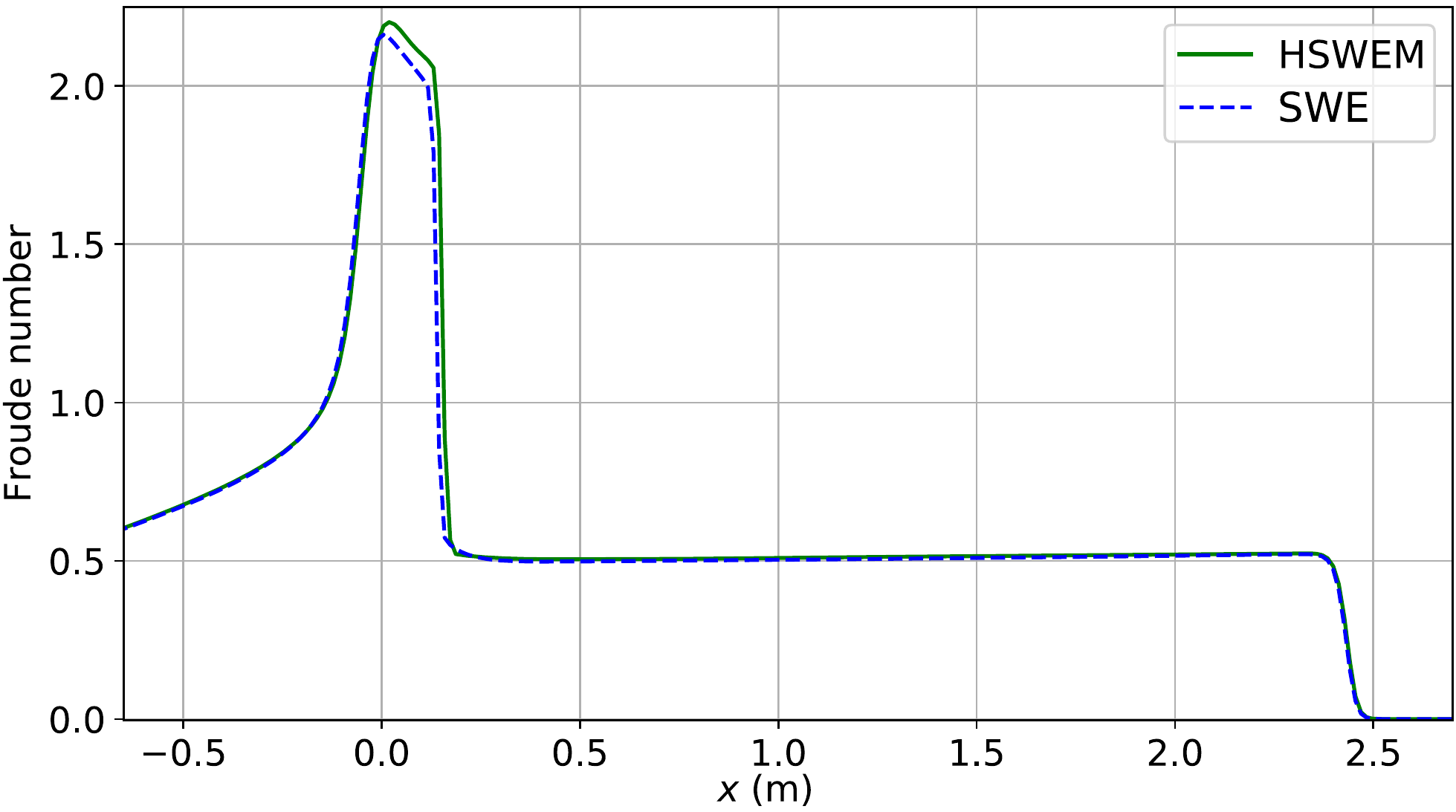}
	\end{center}
	\caption{\label{fig:juez_zb_perfiles_FR} \it{Dam break simulations at time $t=1.5$ s, increasing the Manning friction coefficient $n=0.0365$, with vertical profiles of velocity at $x=-0.5,-0.25,0,0.25,0.5,1,...,2.5$ m, computed with the HSWEM (solid red lines) and SWE (dotted black lines) models.} }
\end{figure}

\section{Conclusions}\label{se:conclusions}
A simple but accurate model for bedload transport, which is obtained as result of applying the moment approximation to the Shallow Water Exner model, has been proposed. The moment approach allows us to recover the vertical velocity profile of the fluid, making it possible to better approximate the velocity close to the bottom. This is the relevant velocity for bedload transport leading to a more accurate definition of the solid transport discharge, and therefore of the sediment transport.

The model considers the Meyer-Peter $\&$ M\"uller definition for the solid transport discharge, together with the Manning friction law at the bottom. This friction law is also considered for the hydrodynamics part of the system. One of the novelties of this work is the derivation of Shallow Water Moment models with a nonlinear friction.

Here, the \textit{Hyperbolic} version of the Shallow Water Moment model, which is always a hyperbolic system (without sediment transport), is used. This model is based on a linearization of the system matrix around linear velocity profiles and can be written as the Shallow Water Moment model with a modification of the nonconservative terms. The combination with sediment transport leads to the new HSWEM model, for which we performed an analysis of the eigenvalues. Most of the eigenvalues are explicitly known, while the remaining eigenvalues correspond exactly to those of the SWE system. 

In order to approximate the model, a numerical scheme based on the IFCP method developed in \cite{gonzalezAguirre:2020} is proposed. We have proposed a coupled discretization to be able to deal with both weak and strong water-sediment interactions. A key point is the approximation of the eigenvalues of the system matrix. Thanks to Theorem \ref{TheoremEV}, this problem is reduced to approximate the eigenvalues of the Shallow Water Exner model. In order to give these approximations, we follow the idea of \cite{cordier:2011}, although showing that this approach fails in some situations, giving three positive eigenvalues (case $u_m>0$), thus leading to the appearance of spurious oscillations in the simulations. A simple modification of this technique has been proposed to ensure that the given approximation verifies one important property for SWE: two eigenvalues have the same sign while the other is of opposite sign. This is crucial for the stability of the scheme. 

Several academic tests have been presented, showing firstly a configuration with a weak fluid-sediment interaction. Concretely, the case of a simple dune that is transported. We have shown that the profiles of the dune and the water surface notably change when the HSWEM model is considered. This is a consequence of the better approximation of the velocity close to the bottom. This is also true even starting from an initially constant profile, since after some time the vertical profile of velocity will be not be constant any more. Secondly, dam-break configurations, where a strong water-sediment interaction occurs, have been showed. The sediment deposits approximated with the HSWEM model exhibits a more realistic shape than with the SWE model, namely at the front where the typical parabolic profile of velocity is recovered. We have also compared our results with laboratory experiments and showed that the results of both models (SWE and HSWEM) only slightly differ. This is due to the fact that the laboratory experiments used are such that the friction force is not large enough to produce a vertical structure in these fast and short time simulations. As it is shown, for increased friction coefficients or larger times, the differences are more noticeable. Interestingly, the major differences take place in the areas where the Froude number is also large, namely where there is a supercritical regime. 

Future works will consider adding erosion-deposition effects to this model, which could notably improve the results when comparing with laboratory experiments where erosion of the bed is important as it is usually the case in dam break problems. Another direction for further research is the application of the filtered moment model in \cite{Fan:2020}, to further reduce the complexity of the moment model and speed up simulation time.

\section*{Acknowledgements}
This research has been partially supported by the Spanish Government and FEDER through the research project RTI2018-096064-B-C2(1/2) and the European Union's Horizon 2020 research and innovation program under the Marie Sklodowska-Curie grant agreement no. 888596. J. Koellermeier is a postdoctoral fellow in fundamental research of the Research Foundation - Flanders (FWO), funded by FWO grant no. 0880.212.840.

\appendix
\section{Proof of HSWEM characteristic polynomial}
\label{app:charPoly}

\begin{theorem}
    \label{appTheoremEV}
    The HSWEM system matrix $\bd{A}_H(\bd{W}) \in \mathbb{R}^{(N+3)\times(N+3)}$ \eqref{e:HSWEM} has the following characteristic polynomial
    \begin{equation*}
        \chi_{A} (\lambda) = \left[(-\lambda)\left( \left(\lambda - u_m\right)^2 - \left(gh+\a_1^2\right) \right) + gh (\dqh + \lambda \dqq + 2\alpha_1 \dqq) \right] \cdot \chi_{A_2} (\lambda - u_m),
    \end{equation*}
    where the matrix $A_2 \in \mathbb{R}^{N \times N}$ is defined as follows

    \begin{equation}
        A_{2}=\left(
        \begin{array}{ccccc}
            & c_2   &       & \\
        a_2 &       & \ddots& \\
            & \ddots&       & c_N  \\
            &       & a_N   &  \\
        \end{array}
        \right),
    \end{equation}
    with values $c_{i}=\frac{i+1}{2i+1}\alpha_1$ and $a_i=\frac{i-1}{2i-1}\alpha_1$ the values above and below the diagonal, respectively, from \eqref{e:HSWEM}.
\end{theorem}

\begin{proof}
    The proof closely follows the proof of the characteristic polynomial of the HSWM system matrix in \cite{koellermeier:2019}, which is extended to include the additional sediment transport. 
    
    We write $\bd{A}_H(\bd{W}) = \widetilde{\bd{A}_H}(\bd{W}) + u_m \bd{I}$ and $\widetilde{\lambda}=\lambda - u_m$, so that we can compute the characteristic polynomial using
    \begin{eqnarray*}
        \chi_{A} (\lambda)  &=& \det \left(\bd{A}_H(\bd{W}) - \lambda \bd{I} \right) \\
        &=& \det \left( \widetilde{\bd{A}_H}(\bd{W}) - (\lambda - u_m) \bd{I} \right)\\
        &=& \det \left( \widetilde{\bd{A}_H}(\bd{W}) - \widetilde{\lambda} \bd{I} \right)\\
        &=& \left| \widetilde{\bd{A}_H}(\bd{W}) - \widetilde{\lambda} \bd{I} \right|.
    \end{eqnarray*}

    When writing $\widetilde{\bd{A}_H}(\bd{W})$ we use the following notation for conciseness:
    \begin{equation*}
    \begin{array}{c}
        d_1 = g h-u_m^2 - \frac{1}{3} \alpha_1^2,\quad d_2 = \frac{2}{3}\alpha_1,\quad d_3 = -2 u_m \alpha_1,\\[3mm]
        d_4 = 2 \alpha_1, \quad d_5 = -\frac{2}{3}\alpha_1^2,\quad d_6 = gh,\quad a_2 = \frac{1}{3}\alpha_1.
        \end{array}
    \end{equation*}
    
    We now compute the determinant $\left| \widetilde{\bd{A}_H}(\bd{W}) - \widetilde{\lambda} \bd{I} \right|$ by developing with respect to the first row and the second row, subsequently.
    \begin{eqnarray*}
    && \left| \widetilde{\bd{A}_H}(\bd{W}) - \widetilde{\lambda} \bd{I} \right| = \left|
        \begin{array}{ccccccc}
        -\widetilde{\lambda} - u_m  & 1 &  &  &  & &\\
        d_1 & -\widetilde{\lambda} + u_m & d_2 &  & & &d_6\\
         d_3 & d_4 & -\widetilde{\lambda} & c_2  &  &  & \\
         d_5 &  & a_2 & -\widetilde{\lambda} & \ddots  &  & \\
        &  &  & \ddots & \ddots & c_N &\\
         &  &  &  & a_N & -\widetilde{\lambda} & \\
        \dqh & \dqq & \dqq & \hdots & \dqq & \dqq & -\widetilde{\lambda} - u_m\\
        \end{array}
        \right| \\
      &=& (-\widetilde{\lambda} - u_m)
      \left|
        \begin{array}{cccccc}
         -\widetilde{\lambda} + u_m & d_2 &  & & &d_6\\
          d_4 & -\widetilde{\lambda} & c_2  &  & &  \\
           & a_2 & -\widetilde{\lambda} & \ddots  & & \\
        &  & \ddots & \ddots & c_N &\\
         &  &  & a_N & -\widetilde{\lambda}  & \\
        \dqq & \dqq & \hdots & \dqq & \dqq & -\widetilde{\lambda} - u_m\\
        \end{array}
        \right|
        -1 \left|
        \begin{array}{cccccc}
        d_1 & d_2 &  &  & & d_6\\
         d_3 &  -\widetilde{\lambda} & c_2  &  &  & \\
         d_5 & a_2 & -\widetilde{\lambda} & \ddots  & &  \\
        &   & \ddots & \ddots & c_N &\\
         &   &  & a_N & -\widetilde{\lambda}  & \\
        \dqh & \dqq & \hdots & \dqq & \dqq & -\widetilde{\lambda} - u_m\\
        \end{array}
        \right| \\
        &=& (\widetilde{\lambda}^2 - u_m^2) \left|
        \begin{array}{ccccc}
           -\widetilde{\lambda} & c_2  &  &  & \\
          a_2 & -\widetilde{\lambda} & \ddots  & &  \\
       & \ddots & \ddots & c_N &\\
         &  & a_N & -\widetilde{\lambda} &  \\
        \dqq & \hdots & \dqq & \dqq & -\widetilde{\lambda} - u_m\\
        \end{array}
        \right| + (\widetilde{\lambda} + u_m) d_2  \left|
        \begin{array}{cccccc}
          d_4 & c_2  &  & & & \\
           &  -\widetilde{\lambda} & c_3  & & & \\
           &  a_3 & -\widetilde{\lambda}  & \ddots & & \\
        & &\ddots & \ddots & c_N &\\
         &   & & a_N & -\widetilde{\lambda} &  \\
        \dqq & \hdots & & \dqq & \dqq & -\widetilde{\lambda} - u_m\\
        \end{array}
        \right|  \\
        && - d_1  \left|
        \begin{array}{ccccc}
           -\widetilde{\lambda} & c_2  &  &  & \\
          a_2 & -\widetilde{\lambda} & \ddots  & &  \\
        & \ddots & \ddots & c_N &\\
         &  & a_N & -\widetilde{\lambda} &  \\
        \dqq & \hdots & \dqq & \dqq & -\widetilde{\lambda} - u_m\\
        \end{array}
        \right| + d_2  \left|
        \begin{array}{cccccc}
         d_3  & c_2  &  &  & &\\
         d_5  & -\widetilde{\lambda} & c_3  &  & &\\
         & a_3 & -\widetilde{\lambda} & \ddots & & \\
        &  &\ddots & \ddots & c_N &\\
         &   & &a_N & -\widetilde{\lambda} &  \\
        \dqh & \dqq & \hdots & \dqq & \dqq & -\widetilde{\lambda} - u_m\\
        \end{array}
        \right| \\
        && +(-\widetilde{\lambda} - u_m) (-1)^{N+3} d_6\left|
        \begin{array}{ccccc}
          d_4 & -\widetilde{\lambda}  & c_2 & &  \\
           &  a_2 & -\widetilde{\lambda}  & \ddots&  \\
        & & \ddots & \ddots & c_N \\
         & &  & a_N & -\widetilde{\lambda}  \\
        \dqq & \hdots & \dqq & \dqq & \dqq \\
        \end{array}
        \right| - (-1)^{N+3} d_6 \left|
        \begin{array}{ccccc}
         d_3  & -\widetilde{\lambda}  & c_2 &  &\\
         d_5  & a_2 & -\widetilde{\lambda} & \ddots  &\\
        &  &\ddots & \ddots & c_N\\
         &   & &a_N & -\widetilde{\lambda}  \\
        \dqh & \dqq & \hdots & \dqq  & -\widetilde{\lambda} - u_m\\
        \end{array}
        \right| \\
        &=& (-\widetilde{\lambda} - u_m)\big[(-\widetilde{\lambda} - u_m)  \left( (-\widetilde{\lambda} + u_m)  \left|A_2\right| - d_2 d_4  \left|A_3\right| \right) - d_1  \left|A_2\right| + d_2 \left( d_3  \left|A_3\right|  - c_2 d_5 \left|A_4\right| \right) \big]\\
        && +(-\widetilde{\lambda} - u_m) (-1)^{N+3} d_6 \left|
        \begin{array}{ccccc}
          d_4 & -\widetilde{\lambda}  & c_2 & &  \\
           &  a_2 & -\widetilde{\lambda}  & \ddots&  \\
        & & \ddots & \ddots & c_N \\
         & &  & a_N & -\widetilde{\lambda}  \\
        \dqq & \hdots & \dqq & \dqq & \dqq \\
        \end{array}
        \right| - (-1)^{N+3} d_6 \left|
        \begin{array}{ccccc}
         d_3  & -\widetilde{\lambda}  & c_2 &  &\\
         d_5  & a_2 & -\widetilde{\lambda} & \ddots  &\\
        &  &\ddots & \ddots & c_N\\
         &   & &a_N & -\widetilde{\lambda}  \\
        \dqh & \dqq & \hdots & \dqq  & -\widetilde{\lambda} - u_m\\
        \end{array}
        \right|
    \end{eqnarray*}
    
    where $|A_i|=\chi_{A_i} (\lambda - u_m)$, for matrix $A_i$ given by 
    \begin{equation}
        A_{i}=\left(
        \begin{array}{ccccc}
            & c_i   &       & \\
        a_i &       & \ddots& \\
            & \ddots&       & c_N  \\
            &       & a_N   &  \\
        \end{array}
        \right) \in \mathbb{R}^{N+2-i},
    \end{equation}
    
    Now the determinants of the last two matrices containing the sediment transport equation are computed separately. Firstly,
    \begin{eqnarray*}
        \left|
        \begin{array}{ccccc}
          d_4 & -\widetilde{\lambda}  & c_2 & &  \\
           &  a_2 & -\widetilde{\lambda}  & \ddots&  \\
        & & \ddots & \ddots & c_N \\
         & &  & a_N & -\widetilde{\lambda}  \\
        \dqq & \hdots & \dqq & \dqq & \dqq \\
        \end{array}
        \right| &=& d_4 \left|
        \begin{array}{ccccc}
            a_2 & -\widetilde{\lambda}  & c_3 & & \\
        &  a_3 & -\widetilde{\lambda} & \ddots &  \\
        &  & \ddots & \ddots & c_N \\
         & &  & a_N & -\widetilde{\lambda}  \\
        \dqq & \hdots & \dqq & \dqq & \dqq \\
        \end{array}
        \right| + (-1)^{N+2} b \left|
        \begin{array}{cccc}
          -\widetilde{\lambda}  & c_2 & & \\
            a_2 & -\widetilde{\lambda}  & \ddots & \\
        & \ddots & \ddots & c_N \\
         &  & a_N & -\widetilde{\lambda}  \\
        \end{array}
        \right| \\[4mm]
        &=& d_4 \left|B_2\right| + (-1)^{N+2} \dqq \left|A_2\right|,
    \end{eqnarray*}
    
    for matrix $B_i$ given by 
    \begin{equation}
        B_{i}=\left(
        \begin{array}{ccccc}
            a_i & -\widetilde{\lambda}  & c_{i+1} & & \\
        &  a_{i+1} & -\widetilde{\lambda} & \ddots &  \\
        &  & \ddots & \ddots & c_N \\
         & &  & a_N & -\widetilde{\lambda}  \\
        \dqq & \hdots & \dqq & \dqq & \dqq \\
        \end{array}
        \right) \in \mathbb{R}^{N+2-i}.
    \end{equation}
    
    And secondly,
    \begin{eqnarray*}
        \left|
        \begin{array}{ccccc}
         d_3  & -\widetilde{\lambda}  & c_2 &  &\\
         d_5  & a_2 & -\widetilde{\lambda} & \ddots  &\\
        &  &\ddots & \ddots & c_N\\
         &   & &a_N & -\widetilde{\lambda}  \\
        \dqh & \dqq & \hdots & \dqq  & -\widetilde{\lambda} - u_m\\
        \end{array}
        \right| &=& d_3 \left|
        \begin{array}{ccccc}
            a_2 & -\widetilde{\lambda}  & c_3 & & \\
        &  a_3 & -\widetilde{\lambda} & \ddots &  \\
        &  & \ddots & \ddots & c_N \\
         & &  & a_N & -\widetilde{\lambda}  \\
        \dqq & \hdots & \dqq & \dqq & \dqq \\
        \end{array}
        \right| - d_5 \left|
        \begin{array}{cccccc}
          -\widetilde{\lambda}  & c_2 &  & & &\\
           & a_3 & -\widetilde{\lambda} & c_4  &\\
        &  & a_4 & -\widetilde{\lambda} & \ddots & \\
        &  &  & \ddots & \ddots & c_N\\
         &   & & & a_N & -\widetilde{\lambda}  \\
         \dqq & \hdots & \dqq &\dqq &\dqq & -\widetilde{\lambda} - u_m\\
        \end{array}
        \right| \\[4mm]
        && + (-1)^{N+2} \dqh \left|A_2\right|
        \\[4mm]
        &=& d_3 \left|B_2\right| -d_5 \left( -\widetilde{\lambda} \left|B_3\right| + (-1)^{N+1} \dqq c_2 \left|A_4\right| \right) + (-1)^{N+2} \dqh \left|A_2\right|.
    \end{eqnarray*}
    
    From \cite{koellermeier:2019}, we use the following recursion rule for $\left|A_{i+1}\right|$:
    \begin{equation}
        \left|A_{i+1}\right| = -\frac{1}{a_{i-1}c_{i-1}} \left( \left|A_{i-1}\right| + \widetilde{\lambda} \left|A_i\right| \right)
    \end{equation}
    and we derive analogously 
    \begin{equation}
        \left|B_{i+1}\right| = \frac{1}{a_{i}} \left( \left|B_{i}\right| + (-1)^{N+4-i} \dqq  \left|A_{i+1}\right| \right),
    \end{equation}
    which we will use for $i=4$ to substitute $\left|A_{4}\right|$ and for $i=3$ to substitute $\left|B_{3}\right|$.
    
    After insertion of the recursion rules, the characteristic polynomial reads
    \begin{eqnarray*}
         \left| \widetilde{\bd{A}_H}(\bd{W}) - \widetilde{\lambda} \bd{I} \right| &=& \left|A_{2}\right| \cdot \left[ (-\widetilde{\lambda} - u_m)  \left( -\widetilde{\lambda}^2 - u_m^2 -d_1 + \frac{d_2 d_5}{a_2} \right) - d_6 \left( (\widetilde{\lambda} + u_m) b - \frac{d_5}{a_2} \dqq + \dqh \right) \right] \\
         && + \left|A_{3}\right| \cdot \left[ (-\widetilde{\lambda} - u_m)  \left( (\widetilde{\lambda} + u_m) d_2 d_4 + d_2 d_3 + \frac{d_2 d_5}{a_2} \widetilde{\lambda} \right) \right]\\
         && + \left|B_{2}\right| \cdot \left[ (-\widetilde{\lambda} - u_m) d_6 d_4 (-1)^{N+3} + d_6 d_3 (-1)^{N+4} + d_6 \frac{d_5}{a_2} \widetilde{\lambda} (-1)^{N+4}\right].\\
    \end{eqnarray*}
    
    Insertion of the entries $d_1,\ldots,d_6$ and $a_2$ then yields the surprisingly simple result
    \begin{eqnarray*}
         \left| \widetilde{\bd{A}_H}(\bd{W}) - \widetilde{\lambda} \bd{I} \right| &=& \left|A_{2}\right| \cdot \left[ (-\widetilde{\lambda} - u_m)  \left( \widetilde{\lambda}^2 - gh -\alpha_1^2 \right) + gh \left( (\widetilde{\lambda} + u_m) \dqq +2\alpha_1 \dqq + \dqh \right) \right] \\
         && + \left|A_{3}\right| \cdot 0 + \left|B_{2}\right| \cdot 0.\\
    \end{eqnarray*}    
    
   Going back to the standard notation with $\widetilde{\lambda}=\lambda - u_m$, we finally have
    \begin{equation*}
        \chi_{A} (\lambda) = \left[(-\lambda)\left( \left(\lambda - u_m\right)^2 -gh -\alpha_1^2 \right) + gh (\dqh + \lambda \dqq + 2\alpha_1 \dqq) \right] \cdot \chi_{A_2} (\lambda - u_m),
    \end{equation*}
    which completes the proof.    
\end{proof}

\bibliographystyle{plain}

\bibliography{Biblio}
\end{document}